\documentclass[journal,draftcls, onecolumn, 12pt]{IEEEtran}
\pdfoutput=1

\usepackage[usenames,dvipsnames]{xcolor}

\usepackage[tight,footnotesize]{subfigure}
\usepackage[pdftex]{graphicx}

\usepackage{graphics} 
\usepackage{graphicx}
\usepackage{verbatim}
\usepackage{color}
\usepackage{subfigure}

\usepackage{caption}
\usepackage{amsmath} 
\usepackage{amsthm} 
\usepackage{amssymb}  
\usepackage{bm}
\usepackage{bbm}
\usepackage{soul}
\definecolor{Gray}{gray}{0.9}

\newtheorem{theorem}{Theorem}
\newtheorem{lemma}{Lemma}
\newtheorem{corollary}{Corollary}

\usepackage{mathtools}

\usepackage{amsmath}
\makeatletter
\newcommand{\distas}[1]{\mathbin{\overset{#1}{\kern\z@\sim}}}%
\newsavebox{\mybox}\newsavebox{\mysim}
\newcommand{\distras}[1]{%
	\savebox{\mybox}{\hbox{\kern3pt$\scriptstyle#1$\kern3pt}}%
	\savebox{\mysim}{\hbox{$\sim$}}%
	\mathbin{\overset{#1}{\kern\z@\resizebox{\wd\mybox}{\ht\mysim}{$\sim$}}}%
}

\def\RR{\mathbb{R}}
\def\P{{ Pr}}  
\def\E{{\rm E}}  
\def\Var{{\rm Var}}  

\begin{document}

\title{Analysis of the Leakage Queue: \\ A Queueing Model for Energy Storage Systems with Self-discharge} 
\author{Majid Raeis*, 
         Almut Burchard**, J\"{o}rg Liebeherr* \\[20pt] 
        * Department of ECE,   University of Toronto,  Canada. \\[-10pt]
        ** Department of Mathematics,   University of Toronto, Canada. \\[-5pt]
	 	E-mail:   m.raeis@mail.utoronto.ca , almut@math.toronto.edu, jorg@ece.utoronto.ca.
        }%

\setcounter{page}{1}
\maketitle
\thispagestyle{plain}
\pagestyle{plain}

\begin{abstract}
Energy storage is a crucial component of the smart grid, since it provides 
the ability to buffer transient fluctuations of the energy supply from
renewable sources.
Even without a load, energy storage systems experience a reduction of the
stored
energy through {\it self-discharge}.  In some storage technologies, the
rate of
self-discharge can exceed 50\% of the stored energy per day.
In this paper, we  investigate the self-discharge phenomenon in energy
storage
using a queueing system model, which we refer to as {\it leakage queue}.
When the average net charge is positive, 
we discover that the leakage queue operates in one of two regimes: a
leakage-dominated regime and a capacity-dominated regime. 
We find that in the leakage-dominated regime,  the stored energy stabilizes at a point 
that is below the storage capacity. 
Under suitable independence assumptions for energy supply and demand, the stored energy in this regime closely follows a normal distribution. 
We present two methods for computing probabilities of underflow and
overflow at a leakage queue. 
The methods are validated
in a numerical example where the energy supply resembles a wind energy source. \end{abstract}


\section{Introduction}
\label{sec:intro}

With their ability 
to absorb the intermittency and uncertainty in renewable energy generation, 
energy storage systems  facilitate the integration of  
renewable energy sources into the grid. 
There exists a wide variety of energy storage technologies, each offering a trade-off with 
regards to storage capacity per unit of volume (energy density), 
delivered power per unit of volume (power density), scalability, 
and others  \cite{dekka2015survey}. 
For instance, compressed air  energy storage systems have a  high capacity but a low power density,  which makes them suitable for long-term storage applications that do not need a fast response time.
On the other end of the spectrum,  supercapacitors have a high power density, but limited storage capacity, which serves applications that require fast responses. Hybrid energy storage systems 
 seek to combine the advantages of different storage technologies in order to meet the need of 
 a specific application \cite{Survey5}. 
Due to the generally high cost of energy storage technologies 
-- 
the price of energy storage can be thousands of dollars per kilowatt hour 
-- 
the economic viability of an energy storage system crucially depends on  properly dimensioning 
the storage size.\footnote{In addition to the price, the cost of energy storage 
also takes into account other factors such as the 
 number of recharge cycles over the lifetime of a unit, the amount of energy that can be 
withdrawn in a single 
recharge cycle ({depth of discharge}), and ancillary costs.}
 Over-provisioning of an energy storage system unnecessarily increases costs, 
while under-provisioning may render it ineffective.  
The need for tools to dimension energy storage systems has 
motivated the development of analytical methods. 
By modelling energy storage systems as  finite-capacity queueing systems with stochastic arrivals (energy supply) and departures (energy demand), the vast 
queueing theory literature becomes available to the problem of storage sizing. 
However, a closer inspection reveals that  energy storage systems are not automatically 
a good fit for a queueing theory analysis.  
For one, arrivals and departures in queueing analysis are often 
expressed as point processes that track arrival and departure events. 
Energy supply and demand, on the other hand, are better characterized by fluid-flow processes.  
Also, many queueing theory methods have been developed for job shop manufacturing and 
communication networks, where buffer provisioning is primarily concerned with preventing 
overflows. Here, it is generally required that the average service rate 
exceeds the average arrival rate. 
Differently, in energy storage the overriding concern is the prevention of buffer underflows 
(`empty batteries'), and, consequently, 
the average arrival rate generally exceeds the average service rate. 

Different from systems usually analyzed by queueing theoretic methods, the stored 
energy may shrink over time even if the system is inactive. The rate of 
leakage of the stored energy, referred to as {\it self-discharge}, results 
from chemical reactions, loss of thermal or kinetic energy, and other factors. 
It can be as low as 1\% of the total charge per month 
for lithium-ion batteries \cite{bat_overview}, and may exceed 50\% per day for 
flywheels \cite{survey3}. 
Interestingly, there is no queueing analysis in the literature 
that accounts for the impact of self-discharge. 
The lack of analytical models for systems with self-discharge was made evident 
in recent performance studies of energy storage \cite{flywheel,regulation,ytrp,yjod,esd,joint-yang}, which 
resorted to optimization methods when accounting for self-discharge. 
The main complication in a performance analysis of storage with self-discharge is 
that self-discharge adds a deterministic or stochastic process 
that runs concurrently with the conventional arrival and service processes. 
Since the quantity of self-discharge depends on the amount of stored energy, the self-discharge rate is not an independent process, but is coupled with the supply and demand 
processes. 

\smallskip
In this paper we study the dynamics of queueing systems that model energy storage systems  
with self-discharge. 
For example, %
how does the self-discharge process interact with the processes for energy supply and 
energy demand? 
Is it possible to identify parameter regions where a storage system 
has a desirable behavior, that is, where overflows and underflows are rare events? 
We consider a  queueing model,  referred to as {\it queue with leakage} or {\it leakage queue}, where supply and demand are governed by 
stochastic processes and, additionally, in each time slot, the content of the queue is reduced 
by a factor $\gamma$, with $0 < \gamma < 1$. 
We consider  leakage queues with finite and infinite capacity. 
Under the assumption of an on average positive net charge, i.e.,  the average supply exceeds 
the average demand, we make a 
number of discoveries, some of which are quite surprising: 
\begin{itemize}
\item We find that a leakage queue with finite capacity has parameter regimes with distinct behaviors, which require different approaches for an analysis. 
\item In the parameter regime where the queue is rarely full or empty, a leakage queue 
is well approximated  by an idealized 
infinite-capacity system that permits stored energy to become negative. 
We find that the distribution of the stored energy in this regime is close to Gaussian. 

\item A leakage queue with $\gamma >0$ is stable for arbitrary supply 
distributions with a finite average. Stability is maintained for systems with 
arbitrarily large capacity.  
Moreover, convergence to the stable steady state occurs exponentially fast. 

\item The steady-state filling level can be precisely determined in a simple expression that 
only requires the averages of the supply and demand processes, as well as $\gamma$.

 
\medskip
\end{itemize}
The goal of our study is to gain insight into the self-discharge phenomenon 
occurring in energy storage systems. 
To isolate the self-discharge effects to a maximum degree, we consider a minimalist  
system model that only leaves the input, output, and leakage processes in place. 
Consequently, our model is not a high-fidelity model for energy storage or a 
particular  storage technology. Also, we do not account for the 
potentially complex interactions between users and utilities 
in demand side management. 
Lastly, we do not address the scale at which energy storage is deployed, that is, whether it 
is utility scale, community scale, or household scale. (We note that  our numerical examples use parameter ranges 
that apply to residential households.)
By eliminating these factors and by considering a bare-bones model of energy storage, 
we are able to  observe and quantify previously unreported dynamics 
in energy storage systems. 


The remainder of the paper is structured as follows. 
In Section~\ref{sec:related}, we discuss the literature on the analysis of energy storage systems. 
In Section~\ref{sec:dynamics}, 
we explore the dynamics of queueing systems with leakage. In Section~\ref{sec:stability},  
we establish the stability of the leakage queue under a broad set of assumptions.
In Section~\ref{sec:bounds},  we present two analytical approaches, each applicable to 
a specified regime of parameters, for deriving overflow and underflow probabilities. In 
Section~\ref{sec:numerics},  we evaluate our analysis with random processes 
that mimic the behavior of a renewable energy source. 
We present conclusions in Section~\ref{sec:concl}.

\section{Related Work}
\label{sec:related}
Energy storage plays a major role in many aspects of the smart grid, 
and, consequently, there is a extensive literature on their analysis. 
The electrical grid requires that power generation and demand load are continuously balanced. 
This becomes more involved with time-variable renewable energy sources and 
storage systems absorbing the variations from such sources. 
Smart grid approaches that take the perspective of a utility operator 
are concerned with placement, sizing, and control of energy storage systems 
with the goal to optimally balance power~\cite{elgamal, sun1,sun2}, 
reduce power generation costs \cite{hassibi}, or operational costs \cite{tassiulas-grid-11}.  
Works in this area are frequently formulated as optimal control or optimization 
problems, with the objective to devise distributed algorithms that achieve 
a desired operating point. 

Demand side management  \cite{dsmsurvey} takes the perspective of an energy user, 
and broadly refers to measures that encourage users to become more 
energy efficient. As one form of demand side management, demand response refers to methods for short-term 
reductions in energy consumption. By creating incentives to users,  
demand response  seeks to match elastic demands with fluctuating  renewable energy sources. 
In \cite{Low-2011-DemandResponse,samadi2010optimal}, demand response is formulated 
as a utility maximization problem where dynamic pricing incentivizes individual 
users to benefit the overall system. 
Studies on demand response apply a wide range of methods, from coordination between 
appliances \cite{coordinating-2010}, bounds on prediction errors \cite{wierman-dr}, 
and game-theoretic approaches~\cite{game-dr-2010}. 

Performance analysis of energy storage systems 
intends to support dimensioning of storage by providing metrics such as 
overflow and underflow probabilities, and the amount of stored energy in the steady state. 
Since detailed models of the circuit or electrochemical processes in an energy storage system, as given in
\cite{chen2006accurate,li_ion,leadacic1992}, 
are not analytically tractable, energy storage systems are generally 
described by  abstract models.
Differential or difference equations have been used for detailed descriptions of 
the energy evolution in lithium-ion batteries \cite{li_ion} and flywheels 
\cite{flywheel, regulation}.  
The suitability of queueing theory for 
analyzing the dynamics of energy storage has been pointed out in \cite{dualkeshav}. 
Interestingly, queueing theory was applied in the 1960s 
for analyzing storage properties of water reservoirs \cite[Chp. III.5]{cohen}, 
and the fluid-flow analysis of queueing systems was known as `dam theory' \cite{damtheory}. 

More recently, a fluid-flow interpretation of queueing theory, known  as 
`network calculus' \cite{Book-LeBoudec}, has been applied to 
energy storage systems. 
A deterministic analysis 
has been used in \cite{leboudec}  to devise battery charging schedules that prevent batteries from running empty. 
Stochastic extensions of the network calculus have been applied to analyze energy storage 
in the presence of random, generally Markovian,  
energy sources  \cite{genset,pnc,jiang}. 
In these works, the evolution of the stored energy is expressed using 
a time-dependent function for the backlog in a finite-capacity queueing system from~\cite{cruz}. 
Recent studies \cite{regulation,ytrp,yjod,esd,joint-yang} have improved the fidelity of energy storage models 
by considering factors such as limited charging and discharging rates, charging and discharging 
inefficiencies, as well as self-discharge. 
In \cite{ytrp}, the self-discharge is modeled by a constant rate function, whereas 
the other works \cite{regulation,yjod,esd,joint-yang} use a proportional leakage ratio as described in Sec.~\ref{sec:intro}. 
Since queueing systems for energy storage systems with proportional self-discharge 
could not be solved analytically, the existing analyses resort to simulation and optimization methods.  
These provide numerical solutions, but do not easily give insight into 
parameter regimes and basic tradeoffs.

\section{A Queueing Model for Energy Storage with Self-Discharge} \label{sec:dynamics}

\begin{figure}[!t]
\centering
\includegraphics[width=4in]{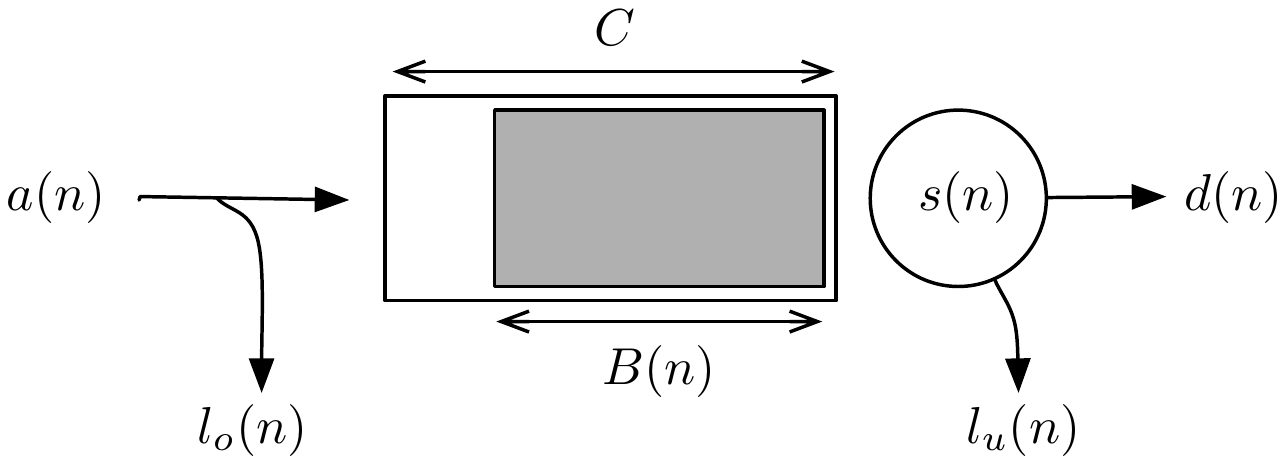}

\centering 
\begin{tabbing}
xxxxxxxxxxxxxxxxx\=xxxxx\= \kill 
\>$a(n)$: \> energy supply ({\it arrivals}) \\ 
\>$s(n)$: \>energy demand ({\it service}) \\ 
\>$d(n)$: \>serviced demand ({\it departures}) \\ 
\>$B(n)$: \>stored energy ({\it backlog}) \\ 
\>$l_o(n)$: \>energy waste ({\it overflow}) \\ 
\>$l_u(n)$: \>energy loss ({\it underflow}) \\ 
\>$C$: \>storage capacity ({\it buffer size}) 
\end{tabbing}
    \caption{Queueing model of energy storage system.} 
    \label{fig:buff_org}
\end{figure}

We model an energy storage system as a finite queueing system, as shown in  
Fig.~\ref{fig:buff_org}. The arrivals to the system consist of  a time-varying  
energy supply from energy sources, the service process consists 
of the  time-varying energy demand from customers, and the departures 
are the serviced demand. 
The stored energy and capacity, which correspond to 
the backlog and capacity in 
a conventional queueing system, are measured in watt hours (Wh). 
In the following, we will use the terms {\it supply} and {\it arrivals}, as well as 
{\it demand} and {\it service} synonymously. 
As a convention, we will employ queueing theory terminology when making 
comparisons to other queueing systems.

\subsection{Dynamics of the leakage queue} \label{subsec:model}
For the purpose of the analysis, we assume that the  
energy processes are discrete-time, fluid flow random processes, where we 
consider deterministic processes as a special case. 
The energy supplied when the storage is at capacity is considered wasted, 
and demand to an empty storage is considered lost.  
We use $a(n), s(n)$, and $d(n)$ to denote the 
energy supply, energy demand, and serviced demand, respectively, in time slot $n$, 
measured in Wh.
We define the {\it net charge } or {\it drift}  of the system in slot $n$, denoted by $\delta (n)$, as the difference 
\[
\delta (n) = a(n) - s(n) \, . 
\]
The amount of energy stored at time slot $n$,  denoted by $B(n)$ and 
referred to as stored energy, 
corresponds to the backlog in the usual terminology of queueing theory. 
Alternative terms in the energy storage literature 
are energy content, state of charge, or battery load. 
The maximum amount of stored energy, referred to as storage capacity, is denoted by $C$. 

We assume that the queue  has a fixed self-discharge ratio $\gamma$ 
with $0 \le \gamma < 1$, with the interpretation that, the stored energy 
at time~$n$, $B(n)$, shrinks to $(1- \gamma) B(n)$ by time slot $n+1$. 
The case $\gamma =0$ refers to a system without self-discharge. 
Since the leakage ratio frequently appears in the form $1- \gamma$, 
we define the complementary leakage ratio $\bar{\gamma}$ as 
\[
\bar{\gamma} = 1- \gamma \, . 
\]
Then, the energy evolution of the storage system can be described in terms of a recursive 
equation \cite{yjod, esd}   by 
\begin{equation}
 B(n)=\min\{[\bar{\gamma} B(n-1)+\delta(n)]^+,C\} \, , 
\label{eq:fin_buf_nonrec}
\end{equation} 
where we use the notation $[x]^+= \max \{x, 0\}$.
We refer to a queueing system with this dynamics as a queue with leakage or leakage queue. 
Descriptions of energy storage systems generally use a fixed self-discharge ratio, 
even though the self-discharge may depend on the 
amount of stored energy, temperature, or other factors. 
Here, the fixed self-discharge ratio represents a long-term average  \cite{survey3}. 

There exist other types of queueing systems where admittance, service, or 
sojourn time are functions of the backlog (stored energy), e.g., 
a {\it G/G/}$\infty$ queue, a queue with 
discouraged arrivals \cite[Chp. 3.3]{Book-Kleinrock}, or a 
reneging queue \cite{reneg-balk2}. 
To see how the leakage queue differs, let us consider 
discrete-time fluid flow versions of these queues. Additionally 
ignoring overflows and underflows, the change of the backlog of the different queues 
in a time slot is given by 
\begin{center}
\begin{tabular}{l l}
Leakage queue: & \hspace{-3pt}$B(n) - B(n-1) = a(n) - s(n) - \bar\gamma B(n-1)$, \\
{\it G/G/}$\infty$ queue: & \hspace{-3pt}$B(n) - B(n-1) = a(n) - s(n) B(n-1)$, \\
Discouraged arriv.: & \hspace{-3pt}
$ B(n) - B(n-1) = \tfrac{a(n)}{B(n-1)+1}- s(n) $, \\
Reneging queue: & \hspace{-3pt}$B(n) - B(n-1) = a(n) \left(1 - \tfrac{B(n-1)}{C}\right) - s(n)$. 
\end{tabular}
\end{center}
Compared to a {\it G/G/}$\infty$ system, the leakage queue has an additional 
process. Compared to queues with discouraged arrivals or reneging, the 
backlog in the leakage queue does not throttle arrivals. 

\smallskip
If time units are expressed in hours, a self-discharge of $5$\% per day for a full battery corresponds to 
a leakage ratio of $\gamma = 0.0021$. This follows since the leakage in a day is $1 - \bar\gamma^{24}$. 
Likewise, we have the correspondences 
\begin{align*}
10\% \text{ per day} \ \ & \sim \ \  \gamma = 0.0044 \, ,  \\
20\% \text{ per day} \ \ & \sim \ \  \gamma = 0.0093 \, ,  \\
50\% \text{ per day} \ \ & \sim \ \  \gamma = 0.0285 \, . 
\end{align*}
Note that the leakage ratios depend on the length of the time slot. 

We use $l_o(n)$ and $l_u(n)$ to denote the overflow and underflow processes, respectively, at 
the storage system. 
In the context of energy storage, $l_o(n)$ is often referred to as the waste of power 
and $l_u(n)$ is referred to as the loss of power. 
The processes are given by 
 \begin{align}
\begin{split}
 l_o(n)&=[\bar{\gamma} B(n-1)+\delta(n)-C]^+ \, , \\
 l_u(n)&=[-\bar{\gamma} B(n-1)-\delta(n)]^+ \, .
\end{split}
\label{eq:define-loss-waste}
 \end{align}

The expression in Eq.~\eqref{eq:fin_buf_nonrec} can be refined to consider other 
pertinent features of an energy storage system. 
For example, some types of storage, such as lithium-ion batteries, 
perform better 
if they are not fully charged. The depth of discharge 
$DoD$ refers to the maximum 
level to which a battery should be charged, expressed as a 
percentage of the storage capacity. 
The charging rate $\alpha_c$ and the discharging rate $\alpha_d$ refer to bounds on the maximum 
power at which the storage can be charged and discharged. 
The charging efficiency~$\eta$ of the storage expresses the amount of energy that is lost 
in the charging process. 
In \cite{ytrp}, these factors are taken into account 
and the evolution of the stored energy is represented by Eq.~\eqref{eq:fin_buf_nonrec}, where 
$C$ is replaced with $C\times DoD $,  and $\delta(n)$ is set to 
 \begin{align}
\delta(n)=\min\bigl\{[\delta(n)]^+,\alpha_c \bigr\} \eta 
-\min\bigl\{[\delta(n)]^+,\alpha_d \bigr\} \,.
\label{eq:fin_buf_nonrec_but}
\end{align} 
As shown in \cite{ytrp}, by defining modified processes for supply and demand  
it is possible to express Eq.~\eqref{eq:fin_buf_nonrec_but} in terms of a simpler energy 
evolution similar to Eq.~\eqref{eq:fin_buf_nonrec}.  
 
We define the bivariate process $\Delta_{\gamma}(m,n)$ as 
 \begin{align*}
 \Delta_{\gamma}(m,n)&= \sum_{k=m+1}^{n} \delta(k) \bar{\gamma}^{n-k} \, .
 \end{align*}
For $\gamma=0$, $\Delta_{\gamma}(m,n)$ is the cumulative net charge process.
With this definition, we present our first result, which is an explicit non-recursive expression for the stored energy in a queue with 
leakage. The result, presented in the next theorem, extends the backlog equation by Cruz and Liu 
for finite-capacity queues~\cite{cruz} to  leakage queues. 

 \begin{theorem} \label{thm:backlog}
Let $B(n)$ be the stored energy in a leakage queue
with finite capacity $C$ and leakage ratio $\gamma$,
as in Eq.~\eqref{eq:fin_buf_nonrec}.
Then
\begin{align}
 B(n) =& \min_{0\leq m \leq n} \bigl\{  \max_{m \leq j \leq n} 
\{C_m \bar\gamma^{n-m} \mathbbm{1}_{j=m} + \Delta_{\gamma}(j,n)
\} \bigr\} \, ,
\label{eq:fin_buff_nonrec}
\end{align}
where $\mathbbm{1}_{j=m}$ is the indicator function
that evaluates to $1$ if $j=m$, and to~$0$ otherwise, and 
$C_m$ is defined as
\[ 
C_m=
\begin{cases} 
B(0)& \text{if } m = 0 \, , \\
C & \text{if } m > 0 \, . 
\end{cases} 
\]
\end{theorem}
Eq.~\eqref{eq:fin_buff_nonrec} implies that the 
effect of the initial charge $C_0$ vanishes as time 
increases. 
By taking $C\to\infty$, we immediately get for a 
leakage queue with infinite capacity  that 
\begin{align*} 
 B(n) =&  \max_{0 \leq j \leq n} \{C_0 \bar{\gamma}^n \mathbbm{1}_{j=0} + \Delta_{\gamma}(j,n) \} \, . 
\end{align*}

 \begin{proof}
We first argue that
\begin{align} 
B(n) &\leq \min_{0\leq m \leq n} \bigl\{  
\max_{m \leq j \leq n} 
\{ C_m\bar\gamma^{n-m}\mathbbm{1}_{j=m} + \Delta_\gamma(j,n)\}
\bigr\} \,.
\label{eq:fwd}
 \end{align}
Let $m$ be an arbitrary time slot
with $0\le m\le n$. If $B(j)>0$ for all
$j$ with $m< j\le n$, then
$$B(j) \le \bar\gamma B(j-1) + \delta(j)\,,\quad
\text{for}\ m<j\le n\,, 
$$
which implies
$$
B(n)\le \bar\gamma^{n-m} B(m) + \Delta_\gamma(m,n)\,.
$$
Otherwise, there exists $j$ with $m<j\le n$ such that
$B(j)=0$, which implies
$$
B(n)\le \Delta_\gamma(j,n)\,.
$$
In either case, since $B(m)\le C_m$, it follows that
\begin{align}
B(n)\le 
\max_{m\le j\le n}
\{ C_m\bar\gamma^{n-m}\mathbbm{1}_{j=m} + \Delta_\gamma(j,n)\}\,.
\label{eq:fwd-m}
\end{align}
Since $m$ was arbitrary, and $B(m)\le C_m$,
this establishes Eq.~\eqref{eq:fwd}.

To complete the proof, it suffices
to find one value of $m$ that produces
equality in Eq.~\eqref{eq:fwd-m}.
Choose $m=0$ if $B(j)<C$ for all $j=1,\dots, n$.
Otherwise, choose $m$ to be the index of the
last time slot up to $n$ with $B(m)=C$. 
Since no overflow occurs 
in time slots $j=m+1,\dots, n$,
the recursion in Eq.~\eqref{eq:fin_buf_nonrec}
yields
$$
B(j)=[\bar\gamma B(j-1)+\delta(j)]^+\,,\quad\text{for}\ m<j\le n\,.
$$
Since $B(m)=C_m$ by the choice of $m$, 
it follows that Eq.~\eqref{eq:fwd-m} 
holds with equality.
\end{proof}
 
\subsection{The dual system}\label{subsec:duality}

Numerous analytical methods are 
available for estimating the overflow probability at a buffered link. 
These methods were developed for
applications of queueing theory in telecommunications and 
manufacturing, e.g., \cite{kelly91}. 
However,
in application areas such as multimedia streaming and 
energy storage, underflow is a more serious concern than overflow. 
By developing dual models where the roles
of underflow and overflow events are
switched, the existing know-how for computing overflow probabilities 
can be leveraged 
for the computation of underflow probabilities \cite{dualkeshav,dualravi}. 
We follow this approach by presenting a dual system for a leakage queue. 
Since the dual system is not a physical system, we 
resort to conventional queueing terminology and talk about 
arrivals, service, and backlog. 

\begin{figure}[!t]
    \centering\includegraphics[width=4in]{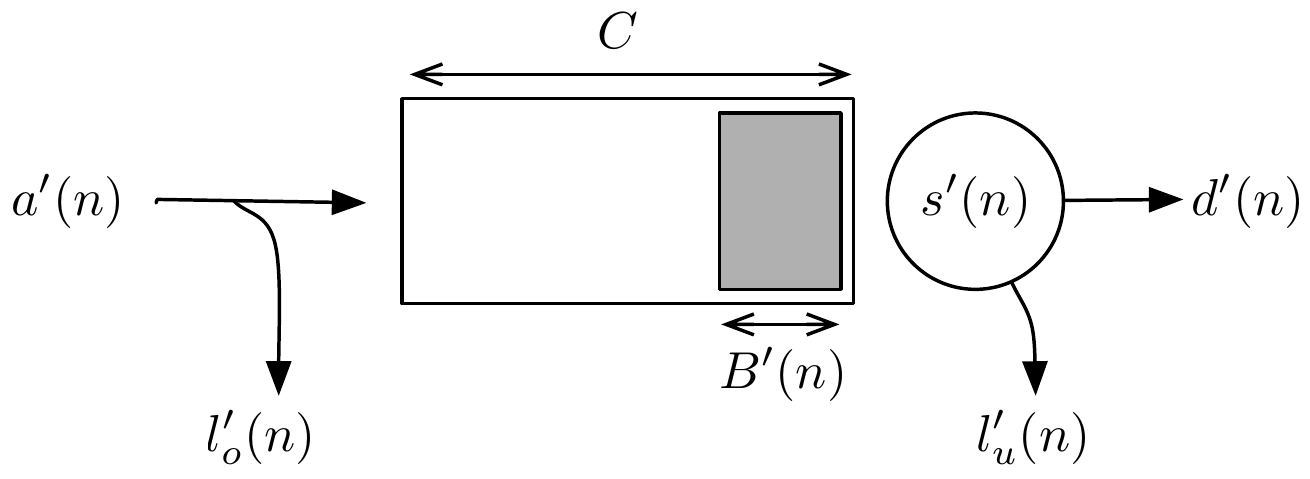}
\centering 
\begin{tabbing}
xxxxxxxxxxxxxxxxxxxx\=xxxxxxxxxxxxxxxxxxxxxx\= \kill 
\>$a'(n)=\gamma C+s(n)$, \\
\> $s'(n)=a(n)$. 
\end{tabbing}
    \caption{Model of the dual system.}
  \label{fig:buff_dual}
\end{figure}

We refer to the leakage queue in Fig.~\ref{fig:buff_org} 
as the original system. The dual 
system is a leakage queue with the same capacity $C$ and leakage ratio $\gamma$.  
Arrivals and service at the dual system, denoted by $a'(n)$ and $s'(n)$,  
are defined as $a'(n)=\gamma C+s(n)$ and $s'(n)=a(n)$, with 
${\delta}'(n)=a'(n)-s'(n)$. 
We denote by $B'(n)$ the backlog process of the dual system. 
The overflow and underflow processes of the dual system,  
denoted by $l'_o(n)$ and $l'_u(n)$, are as in Eq.~\eqref{eq:define-loss-waste}, where 
we replace $\delta(n)$ by $\delta'(n)$ and $B(n)$ by $B'(n)$.
Fig.~\ref{fig:buff_dual} illustrates the queueing model of the dual system. 
With this definition, the backlog $B'(n)$ of the dual system satisfies the recursion 
\begin{align}
B'(n)=\min\bigl\{[\bar\gamma B'(n-1) +\delta'(n)]^+,C\bigr\}\,.
\label{eq:backlog-dual}
\end{align}
Duality of the original and the dual system is established by the following lemma. 
\begin{lemma}
\label{lem:duality}
Given a queue with leakage as shown in Fig.~\ref{fig:buff_org} 
and  the dual system 
shown in Fig.~\ref{fig:buff_dual}. 
If $B(0)+ B'(0)= C$, then the backlog in the 
original system and the dual system satisfy 
\[
B(n)+B'(n)=C
\]
for all $n > 0$.
\end{lemma}
From the lemma it follows immediately that 
$l'_o(n)=l_u(n)$ and $l'_u(n)=l_o(n)$, as long as the 
dual system is properly initialized. 
Hence, we can obtain the underflow probability in the original 
system by computing the overflow probability in the dual system. 

\begin{proof}  We proceed by induction.
The base case is covered by the assumption that
$B(0)+B'(0)=C$.
For the inductive step,
suppose that $B(n-1)+B'(n-1)=C$
for some $n>0$. In particular, $0\le B'(n-1)\le C$.
We rewrite Eq.~\eqref{eq:backlog-dual} 
in terms of $C-B'(n)$ and apply the identity
$$
C-\min\bigl\{[x]^+,C\bigr\}= 
\min\bigl\{[C-x]^+,C\bigr\}
$$
to obtain
\begin{align*}
C-B'(n)& = C- \min\bigl\{[\bar{\gamma} B'(n-1)+\delta(n)]^+,C\bigr\}\\
&= \min\bigl\{[C- \bar{\gamma} B'(n-1)-\delta'(n)]^+, C \bigr\}\\
&= \min\bigl\{[\bar\gamma(C- B'(n-1)) +\gamma C-\delta'(n)]^+, C \bigr\}\,.
\end{align*}
Since $C-B'(n-1)=B(n-1)$ by the inductive hypothesis, 
and $\gamma C-\delta'(n)=\delta(n)$,
it follows that
$$
C-B'(n) = \min\bigl\{[\bar \gamma B(n-1)+\delta(n)]^+,C\bigr\}\,.
$$
We conclude with Eq.~\eqref{eq:fin_buf_nonrec} that $C-B'(n)=B(n)$.
\end{proof}

We can exploit the dual system to obtain
an alternate expression for the backlog.
\begin{corollary}
\label{coro:B-ref}
The backlog in a leakage queue with capacity~$C$ and leakage ratio $\gamma$ 
is given by 
\begin{align*}
B(n) = \max_{0\leq m \leq n} \bigl\{  \min_{m \leq j \leq n} 
\{ 
C_0 \bar\gamma^n \mathbbm{1}_{j=m=0}
+ C \gamma^{n-j} \mathbbm{1}_{j>m}
+ \Delta_{\gamma}(j,n)
\} \bigr\}  \, . 
\end{align*}
\end{corollary}

\begin{proof}
Write $B(n)=C-B'(n)$ and apply Theorem~\ref{thm:backlog} 
to the dual system.
\end{proof}

\begin{figure}[t]
  \centering
	\subfigure[$C=40$ kWh.]{
 	\includegraphics[width=5in]{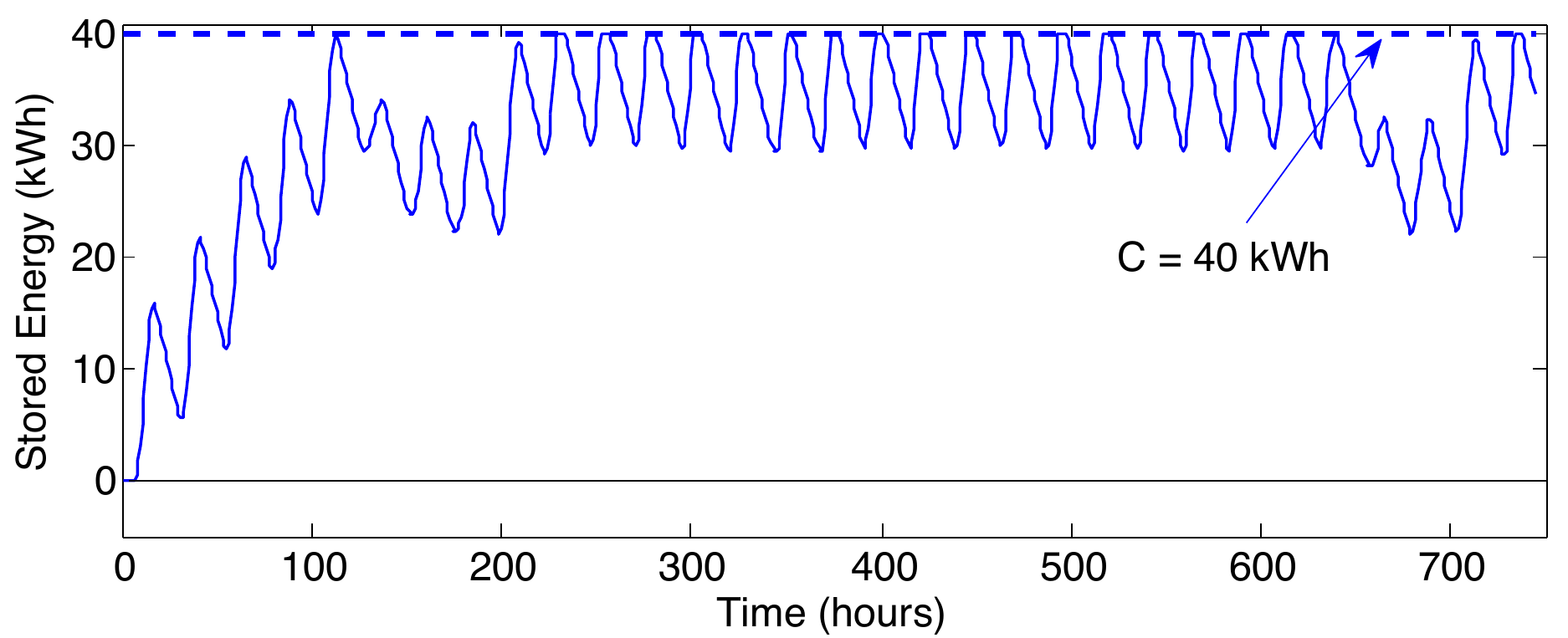}
	\label{fig:LA_solar_2}
	}
	\subfigure[$C=10$ kWh.]{
 	\includegraphics[width=5in]{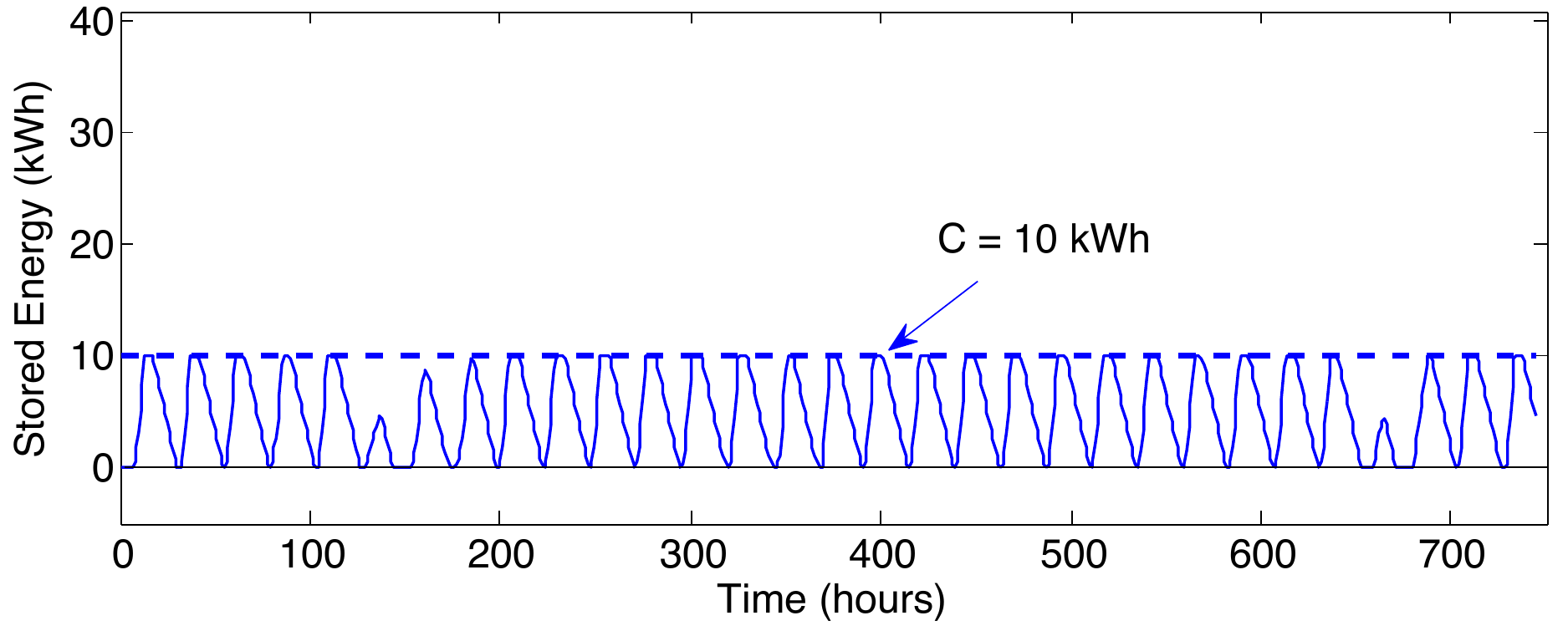}
	\label{fig:LA_solar_1}
	}
\caption{Stored energy without self-discharge with solar power supply.} 
	\label{fig:LA_solar}
\end{figure}
\begin{figure}[t]
  \centering
	\subfigure[$C=40$ kWh.]{
 	\includegraphics[width=5in]{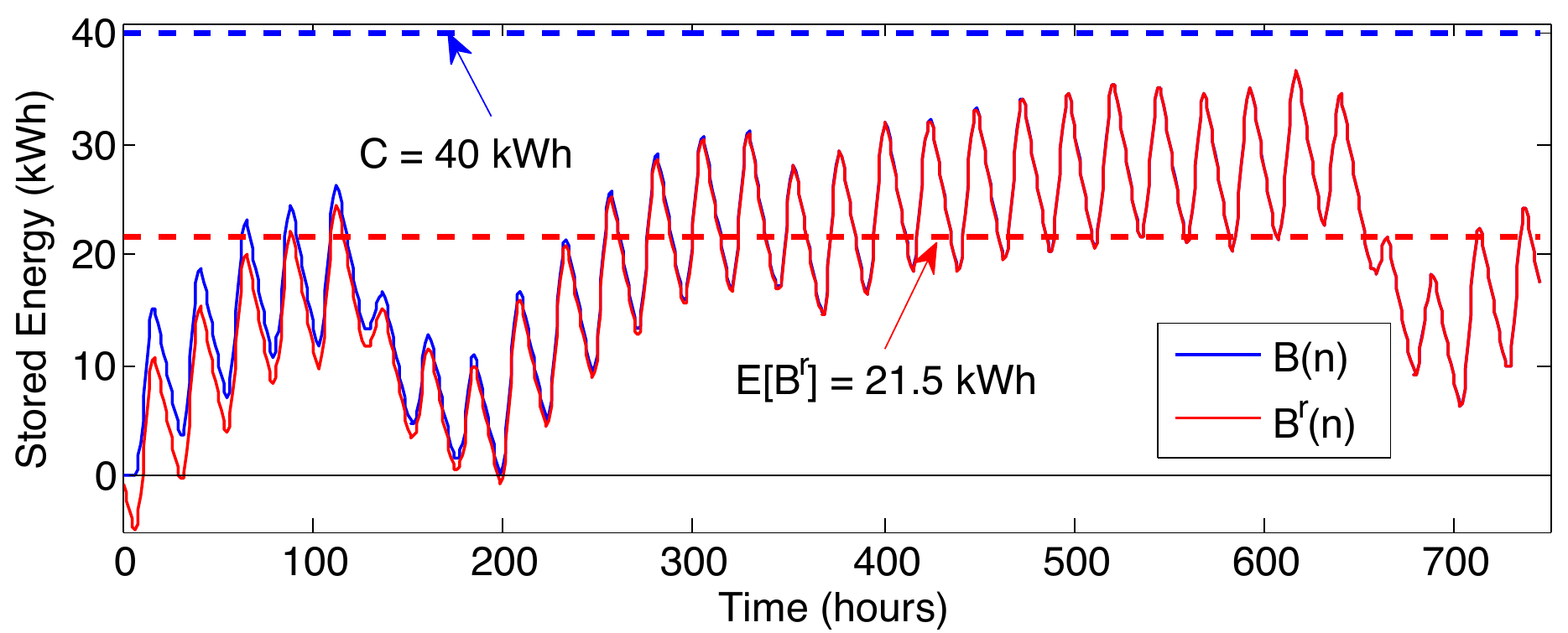}
	\label{fig:LA_solar_leak_40}
	}
	\subfigure[$C=10$ kWh.]{
 	\includegraphics[width=5in]{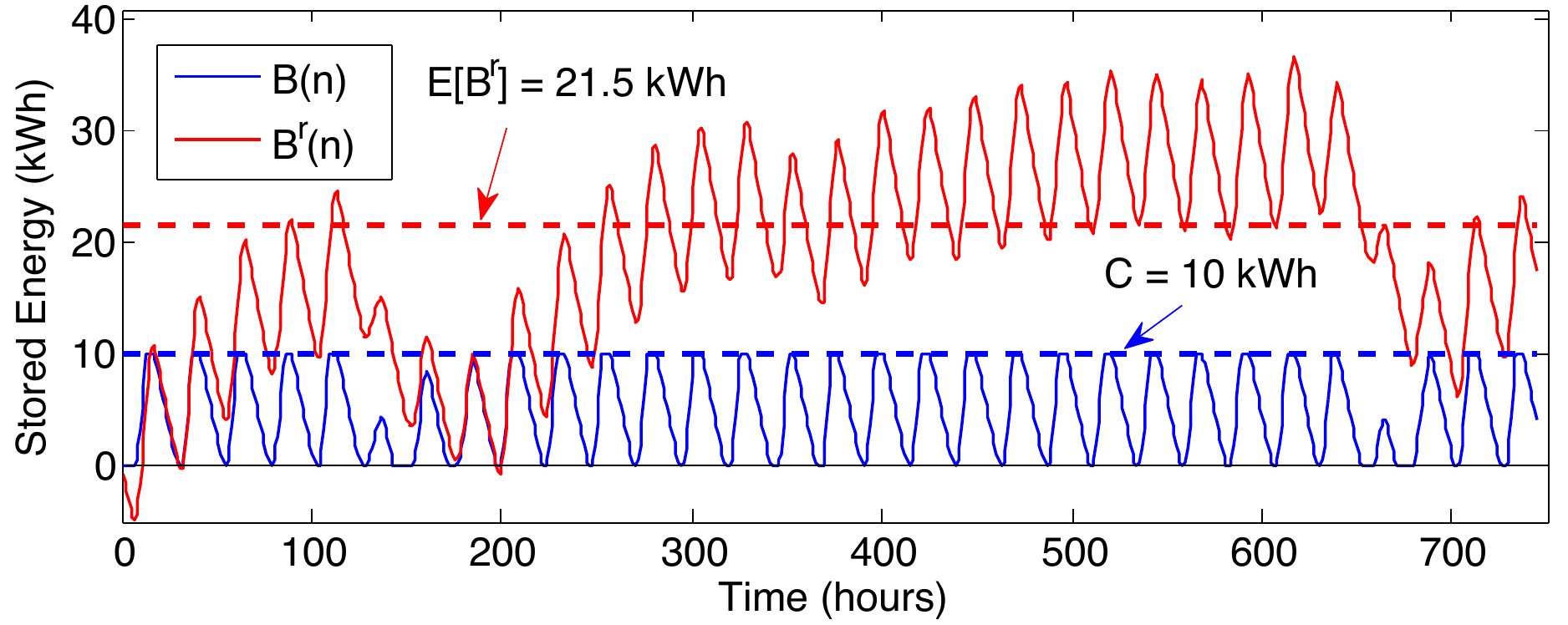}
	\label{fig:LA_solar_leak_10}
	}
\caption{Stored energy with 20\% self-discharge per day with solar power supply.} 
	\label{fig:LA_solar_leak}
\end{figure}

\subsection{Two regimes for the analysis}
\label{subsec:ref_system}

In this subsection, we make 
observations that will prove crucial for 
the analysis of leakage queues. 
Throughout, we will work with a net charge that is positive on average, 
since the storage system will be mostly empty otherwise. 
We find that the leakage queue 
operates in two regimes with 
fundamentally different behaviors. In one regime, 
the stored energy is stable at 
a point below the storage capacity. 
Here, the leakage queue
behaves similarly to a reference system, 
which has  infinite capacity and allows 
the stored energy to become negative. 
In the other regime, the stored energy  is 
generally close to the capacity.
Here, 
the leakage queue is similar to a conventional 
finite-capacity queueing system in overload.

We illustrate the different regimes with the aid of a numerical example 
drawn from an energy storage system with a photo-voltaic (PV) energy source and constant demand. 
We use  the PV energy generation 
for a residential rooftop system, which is based 
on an hourly  data set of the typical solar irradiance in Los Angeles for 
the month of July \cite{sam-data}.  
The resulting solar energy is calculated with the 
System Advisor Model (SAM) software~\cite{sam}, 
where solar panels are scaled so that the  average energy supply  per hour is 
$1$~kWh.
The average demand per hour is assumed to be constant and set to 
$800$~Wh, which is 80\% of the supply. 
This is approximately the average power consumption per household in New Zealand \cite{household}.  

For calibration, we first consider a queue without leakage, that is, 
$\gamma = 0$.  In Fig.~\ref{fig:LA_solar} we depict the energy content 
for systems with storage capacity $C=10$~kWh and $C=40$~kWh, where we assume that the storage 
is initially empty. Observe that the data captures the diurnal pattern of solar energy. 
As expected,  once the storage fills up, the stored energy is always close to the storage capacity. 

Before we discuss how the outcome changes in the presence of self-discharge, we introduce a reference 
system that differs from the leakage queue in two ways. First, the reference 
system  has infinite storage capacity ($C=\infty$). Second, the stored energy is allowed to 
take negative values, where a negative occupancy can be thought of as an energy deficit. 
The dynamics of the stored energy in the reference system, denoted by~$B^r$, 
simplify to 
\begin{align}
B^r(n)=\bar{\gamma} B^r(n-1)+\delta(n) \, . 
\label{eq:def_Br}
\end{align}
Solving the recursion yields the formula
\begin{align*}
B^r(n)=B(0) \bar\gamma^n  + \sum_{m=0}^n \Delta_\gamma(m,n)\,.
\end{align*}
If $\delta (n)$ describes a stationary process with 
$\delta (n) =_{\mathcal{D}} \delta$ for all $n$, where `$=_{\mathcal{D}}$' 
indicates equality in distribution, 
the expected value is 
\[
\E[B^r(n)]=\bar{\gamma} \E[B^r(n-1)]+\E[\delta] \, .
\]
If there exists a steady state for $B^r$, then,  for $n\to\infty$,  
the expected stored energy, denoted by $\E [B^r]$, is given by
\begin{equation}
\E[B^r]=\lim_{n \to \infty}\E[B^r(n)]=\frac{\E[\delta]}{\gamma} \, . 
\label{eq:reference-bstar}
\end{equation}
In the next section, we will prove that $B^r(n)$ always converges
to a steady state, with expected value $\E[B^r]$.

We now re-compute the numerical example from Fig.~\ref{fig:LA_solar} 
with a leakage ratio of 
$\gamma = 0.0093$. With the given supply and demand we obtain that 
$\E[B^r] =21.5$~kWh, where, for the data set,  $\E[B^r]$ 
is the average net charge divided by $\gamma$. 
Fig.~\ref{fig:LA_solar_leak} shows the stored energy 
for the reference system and 
the finite-capacity leakage queue. 
Note that the reference system initially takes negative values for $B^r(n)$. 
In Fig.~\ref{fig:LA_solar_leak_40}, where $C > \E[B^r]$, 
the stored energy in the reference system tracks the 
energy in the finite-capacity leakage queue with a high degree of accuracy. 
In Fig.~\ref{fig:LA_solar_leak_10}, we show the 
results for $C < \E[B^r]$. Here, the 
stored energy in the finite-capacity leakage queue is very different from that of 
the reference system.  In fact, the dynamics of the leakage queue 
resemble that of the finite-capacity queue without leakage 
(shown in Fig.~\ref{fig:LA_solar_1}). 

In the next sections, we will establish that our 
observations in the numerical example 
extend to other supply and demand distributions. 
It turns out that all leakage queues 
with $\gamma > 0$ operate in one of two modes, which we will refer to as 
 {\it capacity-dominated regime} and  {\it leakage-dominated regime}, with 
very different characteristics. 

 \begin{itemize}
 \item {\bf Leakage-dominated regime ($C > \E[B^r]$):}  This regime, 
illustrated in Fig.~\ref{fig:LA_solar_leak_40},  
is characterized by an average stored energy below the 
storage capacity.  
This is unlike a conventional finite capacity 
queueing systems (with $\gamma=0$), where the storage is 
full or close to full for $\E[\delta] > 0$.
 
\item {\bf Capacity-dominated regime ($C < \E [B^r]$):} In 
this regime, illustrated in Fig.~\ref{fig:LA_solar_leak_10},  the stored energy is always close to the storage capacity $C$,  
which necessarily results in a high probability of overflow. 
The system behaves similarly to a conventional 
finite-capacity queueing system 
without leakage.

\end{itemize}

We will study the regimes in detail in Section~\ref{sec:bounds}, where we 
find that we must use different analysis methods for each regime. 

 \begin{figure}[!t]
\centering
\includegraphics[width=5in]{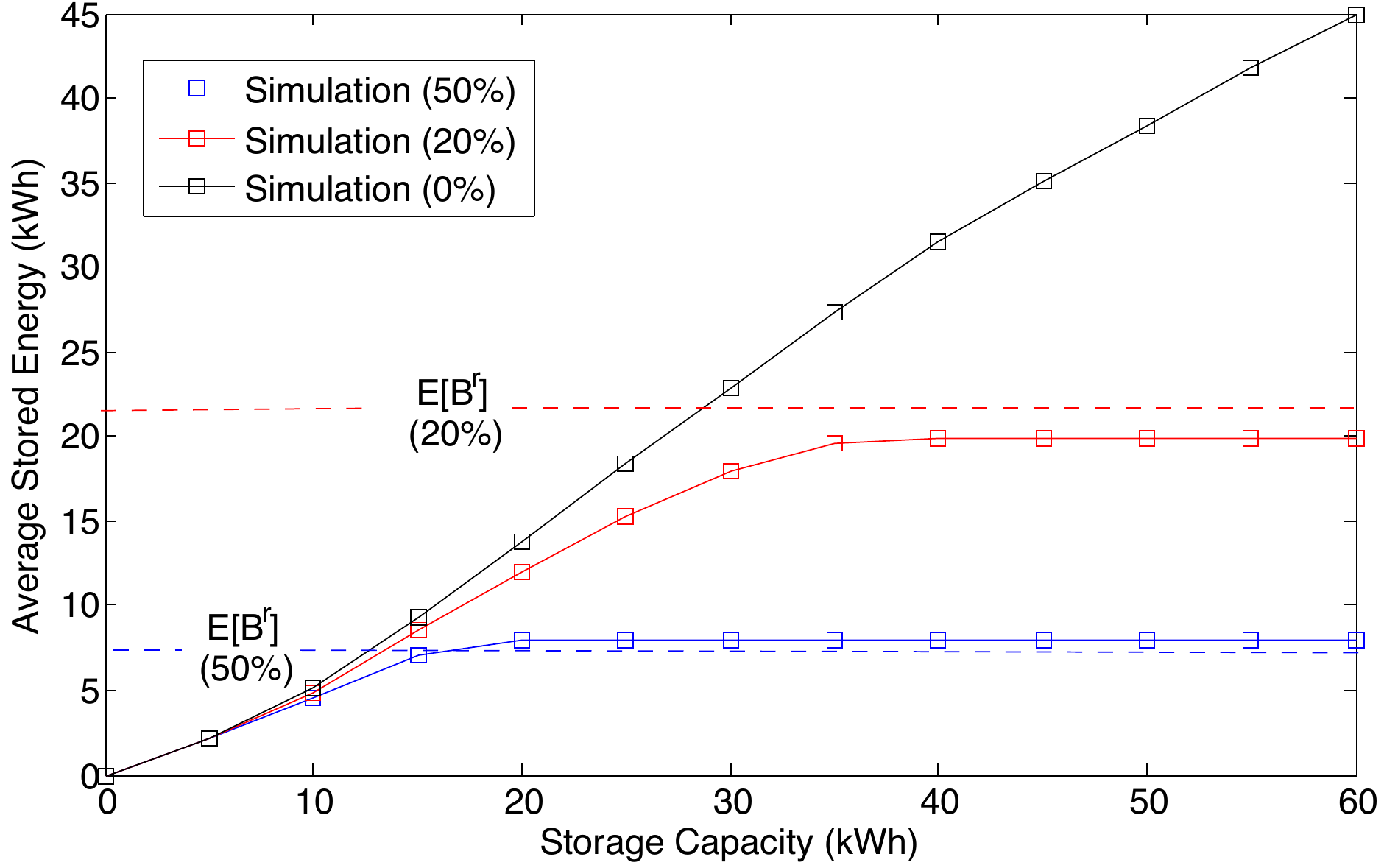}
\caption{Average stored energy as a function of the storage capacity $C$. }

\label{fig:average_backlog}
\end{figure}

An interesting property of leakage queues is that 
increasing the storage capacity much beyond $\E[B^r]$ 
does not result in significant benefits. 
To emphasize this, we consider the same data set as 
used for Figs.~\ref{fig:LA_solar} and~\ref{fig:LA_solar_leak} 
and compute the average stored energy as a function of the storage capacity. 
In Fig.~\ref{fig:average_backlog} we show the 
results for self-discharge ratios of $\gamma = 0, 0.0093, 0.0285$.
Without self-discharge ($\gamma=0$), the average 
stored energy is always close to the capacity.  
For $\gamma>0$, on the other hand, the average storage 
approaches a constant even as the capacity goes to infinity.
The value of this constant
is close to the stored energy in
the reference system. 
To our knowledge, this feature of energy storage systems with self-discharge 
has not received any attention. 

Obviously, our observations of the stability of the leakage queue and 
the characterization of the leakage-dominated regime by the reference system 
are limited to the depicted data set. In the next sections, we will 
try to corroborate our findings for general random processes.


 \section{Stability and Convergence}
\label{sec:stability}

In this section, we prove 
that the leakage queue defined by
Eq.~\eqref{eq:fin_buf_nonrec}
converges for $0<\gamma<1$ 
to a unique steady state as $n\to\infty$.
The net charge $\delta(n)$ for $n = 0, 1, \ldots$ is assumed to 
be an i.i.d. sequence of random variables of finite
expectation $\E[\delta]$.
Assuming that $\E[\delta]>0$, we characterize the steady-state 
distribution of the stored energy,~$B$,
in terms of the steady state of the reference system $B^r$ and
the capacity $C$. In the leakage-dominated regime, where 
$C> E[B^r]$, we prove that the steady state is close to 
that of the reference system. 
In the capacity-dominated regime, where $C < E[B^r]$,
the expected drift in the dual system is
negative. The steady state
resembles that of a stable queue without leakage. 

The stability and convergence results
extend to leakage queues with infinite capacity.
Obviously, a leakage queue of infinite capacity is
always leakage-dominated.
%

\subsection{Stability} 

A leakage queue of finite capacity $C$ is stable
by definition,
since $0 \le B(n) \le C$ for all $n\ge 0$. 
We will derive bounds on the distribution 
of $B(n)$ that do not depend on the value of $C$. 

By the recursive definition 
from Eq.~\eqref{eq:fin_buf_nonrec}, we have 
$
B(n)\le \bar \gamma B(n-1) + [\delta(n)]^+ 
$. 
Solving the recursion, and then using 
the i.i.d. assumption on the drift, we obtain 
\begin{align*}
B(n)&\le  B(0)\bar \gamma^n +\sum_{m=1}^n [\delta(m)]^+
\bar\gamma^{n-m}\\
&=_ \mathcal{D} B(0)\bar \gamma^n +
\sum_{m=0}^{n-1} [\delta(m)]^+\bar\gamma^m=: Y(n)\,,
\end{align*}
In particular,
$$\E[B(n)] \le \E[Y(n)] = \frac{1-\bar\gamma^{n}}{\gamma}
\E\bigl[ [\delta]^+\bigr]\,.
$$
As $n\to\infty$, the random variables
converge to 
$$
Y:= \sum_{m=0}^{\infty} \bar\gamma^m [\delta(m)]^+\,.
$$
By monotone convergence, $Y$ has finite
mean, and is almost surely finite.  By construction,
$
\P(B(n)>x) \le \P(Y>x)
$ 
holds for all $n\ge 0$. Since $B(n)$ is nonnegative,
it follows that the leakage queue is stable.

For the reference system, a similar argument, using dominated convergence,
shows that $B^r(n)$ converges in distribution to
\begin{equation} 
\label{eq:b_r_inf}
B^r : =\sum_{m=0}^\infty \delta(m)\bar\gamma^{m}\,.
\end{equation}
Since this series converges absolutely almost surely,
the reference system is stable as well.
We will refer to $B^r$ as the steady-state distribution
of the reference system. 
Note that the expected value of Eq.~\eqref{eq:b_r_inf} 
is equal to Eq.~\eqref{eq:reference-bstar}, justifying 
our earlier use of the notation~$\E [B^r]$.

\subsection{Convergence to steady state}\label{subsec:convergence} 

We next show that the stored energy in a leakage queue
converges in distribution to a steady state. 
\begin{theorem}
\label{thm:convergence}
Let $\delta(n)$ be an i.i.d. sequence of finite mean,
$C>0$, and $0 < \gamma<1$.
Then the stored energy $B(n)$ in a queue 
with leakage ratio~$\gamma$, storage capacity $C$, and 
drift $\delta(n)$ converges in distribution to a steady state
that does not depend on the initial condition.
\end{theorem}

The key observation is that Eq.~\eqref{eq:fin_buf_nonrec}
defines a contraction in a suitable metric
on the space of probability distributions.
Given two random variables $X_1$, $X_2$,
with cumulative distribution functions (CDF)
$F_1(x)=\P(X_1\le x)$ and $F_2(x)=\P(X_2\le x)$.
Their Kantorovich-Rubinstein distance is defined by
$$
d(F_1,F_2) = \int_{-\infty}^\infty
|F_1(x)-F_2(x)|\, dx\,. 
$$ 
With a slight abuse of notation, we write
$d(X_1,X_2)$ in place of $d(F_1,F_2)$
for the Kantorovich-Rubinstein
distance between the distributions
of $X_1$ and $X_2$.
The following technical lemma provides the necessary estimates.

\begin{lemma}
\label{lem:KR}
Let $X_1$ and $X_2$ be random variables of finite mean.
Then
\begin{enumerate}
\item $d(\alpha X_1,\alpha X_2)=\alpha d(X_1,X_2)$
for every $\alpha>0$.
\item $d([X_1]^+,[X_2]^+)\le d(X_1,X_2)$.
\item $d(\min\{X_1,C\}),\min\{X_2,C\})\le d(X_1,X_2)$
for every $C\in\RR$.
\item $d(X_1+Y,X_2) \le d(X_1,X_2)$ for every random variable $Y$
of finite mean that is independent of $X_1$ and $X_2$.
\end{enumerate}
\end{lemma}

\begin{proof} 
For the first claim, we use that
the CDF of $\alpha X_i$ is
$$
\P(\alpha X_i\le x)=F_i(\alpha^{-1}x)\,,\quad i=1,2
$$
and compute
\begin{align*}
d(\gamma X_1,\gamma X_2) &= 
\int_{-\infty}^\infty
|F_1(\alpha^{-1}x)-F_2(\alpha^{-1} x)|\, dx\\
&=\alpha
\int_{-\infty}^\infty
|F_1(y)-F_2(y)|\, dy\\
&=\alpha \,d(X_1,X_2)\,.
\end{align*}
The next two claims are immediate 
from the facts that
the CDF of $[X_i]^+$
is $F_i \mathbbm{1}_{x>0}$,
and the CDF of
$\min\{X_i,C\}$ is
$F_i \mathbbm{1}_{x\le C}$. For the final claim,
let $\mu$ be the probability distribution
of $Y$. Then the CDF of $X_i+Y$ is 
$F_i*\mu(x)=\int F_i(x-y)\, d\mu(y)$ for $i=1,2$.
Therefore
\begin{align*}
d(X_1+Y,X_2+Y)
&=\int_{-\infty}^\infty
\left|\int_{-\infty}^\infty F_i(x-y)
-F_2(x-y)\,d\mu(y)\right|\, dx\\
&\le \int_{-\infty}^\infty
\int_{-\infty}^\infty |F_i(x-y) -F_2(x-y)|\,d\mu(y)\, dx\\
&=d(X_1,X_2)\,.
\end{align*}
We have used the triangle inequality, and then applied Fubini's theorem.
\end{proof}

\begin{proof}[Proof of Theorem~\ref{thm:convergence}]
Let $\Psi$ be the transformation that maps the distribution of
$B(n-1)$ to the distribution of $B(n)$ according
to Eq.~\eqref{eq:fin_buf_nonrec}. Explicitly,
\begin{equation}
\label{eq:def-Psi}
\Psi(X)=_{\mathcal{D}}
\min\bigl \{[\bar\gamma X+\delta]^+, C\bigr\}\,,
\end{equation}
where $\delta$ is independent of $X$.
By Lemma~\ref{lem:KR}, 
\begin{align*}
d\bigl(\Psi(X_1),\Psi(X_2))
\le d\bigl(\bar\gamma X_1, \bar\gamma X_2\bigr)
= \bar\gamma d(X_1,X_2)\,.
\end{align*}
By Banach's contraction mapping theorem,
$\Psi$ has a unique fixed point, which we denote by
$B$. Moreover, by induction,
$$
d(B(n),B)\le \bar\gamma^n d(B(0),B)\,,
$$
proving convergence to the steady state. 
\end{proof}
The last step of the proof allows us to strengthen 
 Theorem~\ref{thm:convergence}. 
\begin{corollary} 
The convergence of $B(n)$ to the steady state 
in Theorem~\ref{thm:convergence} occurs exponentially 
fast. 
\end{corollary} 

The same argument, shows the convergence of a 
leakage queue of infinite capacity to its 
steady state. The only change is that
Eq.~\eqref{eq:def-Psi} should be
replaced by 
$\Psi(X)=_{\mathcal{D}} [\bar\gamma X+\delta]^+$.

The proof of the theorem also yields
an estimate for the distance of the steady state
of the stored energy from the steady state of
the reference system, given in Eq.~\eqref{eq:b_r_inf}.

\begin{corollary}
\label{coro:B-Br}
The steady state distributions of the stored energy
in the leakage queue and the reference
system satisfy
$$
d(B,B^r)\le \frac{1}{\gamma}\, 
\left(\int_{-\infty}^0 \P(B\le x)\, dx + \int_C^\infty \P(B>x)\, dx
\right)\,.
$$
\end{corollary}

\begin{proof} By Theorem~\ref{thm:convergence},
the stored energy $B(n)$ in the leakage queue converges to
the steady state regardless of the choice of the initial
condition $B(0)$.  Let us use the
steady state of the reference system as an
initial condition for $B(n)$, that is, 
$B(0)=_{\mathcal{D}} B^r$. 
It is a consequence of
the contraction mapping theorem that
the distance to the steady state 
is bounded by
$$
d(B,B(n))\le \frac{\bar\gamma^n}{\gamma}\, d(B(1),B(0))\,.
$$
We set $n=0$ and proceed to estimate
$d(B(1),B(0))$. Due to our choice of the initial state, 
Eq.~\eqref{eq:fin_buf_nonrec} yields
$$
B(1)=_{\mathcal{D}} \Psi(B) = \min\bigl\{[B]^+,C\bigr\}\,,
$$
because $\bar\gamma B^r +\delta=_{\mathcal{D}}B^r$
by the definition of the reference system.
Therefore
$$
d(B(1),B(0)) =
\int_{-\infty}^0 \P(B\le x)\, dx + \int_C^\infty \P(B>x)\, dx\,,
$$
completing the proof.
\end{proof}

\section{Probabilistic Bounds}
\label{sec:bounds}

In this section, we quantify the underflow and overflow probabilities 
$\P (l_u > 0)$ and $\P (l_o > 0)$ at the leakage queue. 
Based on Sec.~\ref{sec:dynamics}, we expect the 
reference system to provide a good approximation for the leakage queue 
in the capacity-dominated regime where the storage capacity 
is large enough to absorb random variations of power supply and demand.
For the capacity-dominated regime, we offer a separate martingale analysis.  
As in Sec.~\ref{sec:stability}, we assume that the 
net charge $\delta(n)$ is an i.i.d. process with 
finite mean. 
We now additionally assume that $\delta(n)$ has finite variance.

\subsection{Gaussian analysis}\label{subsec:bounds-gaussian}

In Sec.~\ref{sec:stability} we showed that the 
stored energy in the reference system $B^r(n)$ has a unique steady state,
given by Eq.~\eqref{eq:b_r_inf}. 
In the special case where $\delta(k)$ follows a normal distribution,
the reference system $B^r$ is also Gaussian, with mean and variance 
in the steady state given by
\begin{align}
 \E[B^r]=\frac{\E[\delta]}{\gamma}, \qquad
 \Var[B^r]=\frac{\Var [\delta]}{1-\bar\gamma^2} \, . 
\label{eq:gauss_prob_1}
\end{align}
Let $B$ be the steady state of the
corresponding leakage queue.
If $C > E[B^r]$, we expect
the stored energy $B$ to be well-approximated by $B^r$,
see Corollary~\ref{coro:B-Br}.
In particular, underflow and overflow probabilities 
should be small, and satisfy
\begin{align}
\label{eq:gauss_approx-lu} 
\P (l_u > 0) & \approx 
\P(B^r<0) =
\Phi \left(\frac{\E[B^r]}{\sqrt{\Var [B^r]}} \right) \, , \\
\label{eq:gauss_approx-lo} 
\P (l_o > 0) &\approx 
\P(B^r>C) = 
\Phi \left(\frac{C-\E[B^r]}{\sqrt{\Var [B^r]}} \right) \,   ,
\end{align}
where $\Phi$ is the standard normal CDF.

We evaluate the accuracy of this approximation for a leakage queue with 
Gaussian net charge, by comparing  
Eqs.~\eqref{eq:gauss_approx-lu} and~\eqref{eq:gauss_approx-lo} 
with simulations of the leakage queue. 
For the simulations, we compute averages over 
multiple repetitions of long simulation runs.
We consider a storage system with a size up to 50~kWh, which covers a 
reasonable range for residential energy storage systems~\cite{tesla}.  
We assume a self-discharge of 20\% per day ($\gamma = 0.0093$)
or 50\% per day ($\gamma = 0.0285$), and also consider a system without self-discharge 
($\gamma = 0$).
The energy supply and demand in a time slot of 
one hour are set so that $\E[a]=1$~kWh and $\E[s]=0.8$~kWh, respectively. 
The standard deviation  is set to $\sigma_\delta = 1$~kWh.

Fig.~\ref{fig:P_underflow_gauss_leakdom} depicts the underflow probability  
computed from Eq.~\eqref{eq:gauss_approx-lu} as a function of the 
storage capacity, and compares it with simulations. 
Since the expression for $\P(B^r<0)$ 
in Eq.~\eqref{eq:gauss_approx-lu} does not depend on $C$, 
the analysis yields a straight line. 
The simulations of a system without leakage 
show that the underflow probability decreases exponentially 
in $C$. 
For systems with leakage,  
on the other hand, 
the underflow probability becomes eventually constant 
in the leakage-dominated region ($C > \E[B^r]$). 
Hence, increasing the storage capacity 
further will not reduce the underflow probability.

\begin{figure}[t]
  \centering
	\subfigure[Energy Loss (Underflow).]{
 	\includegraphics[width=5in]{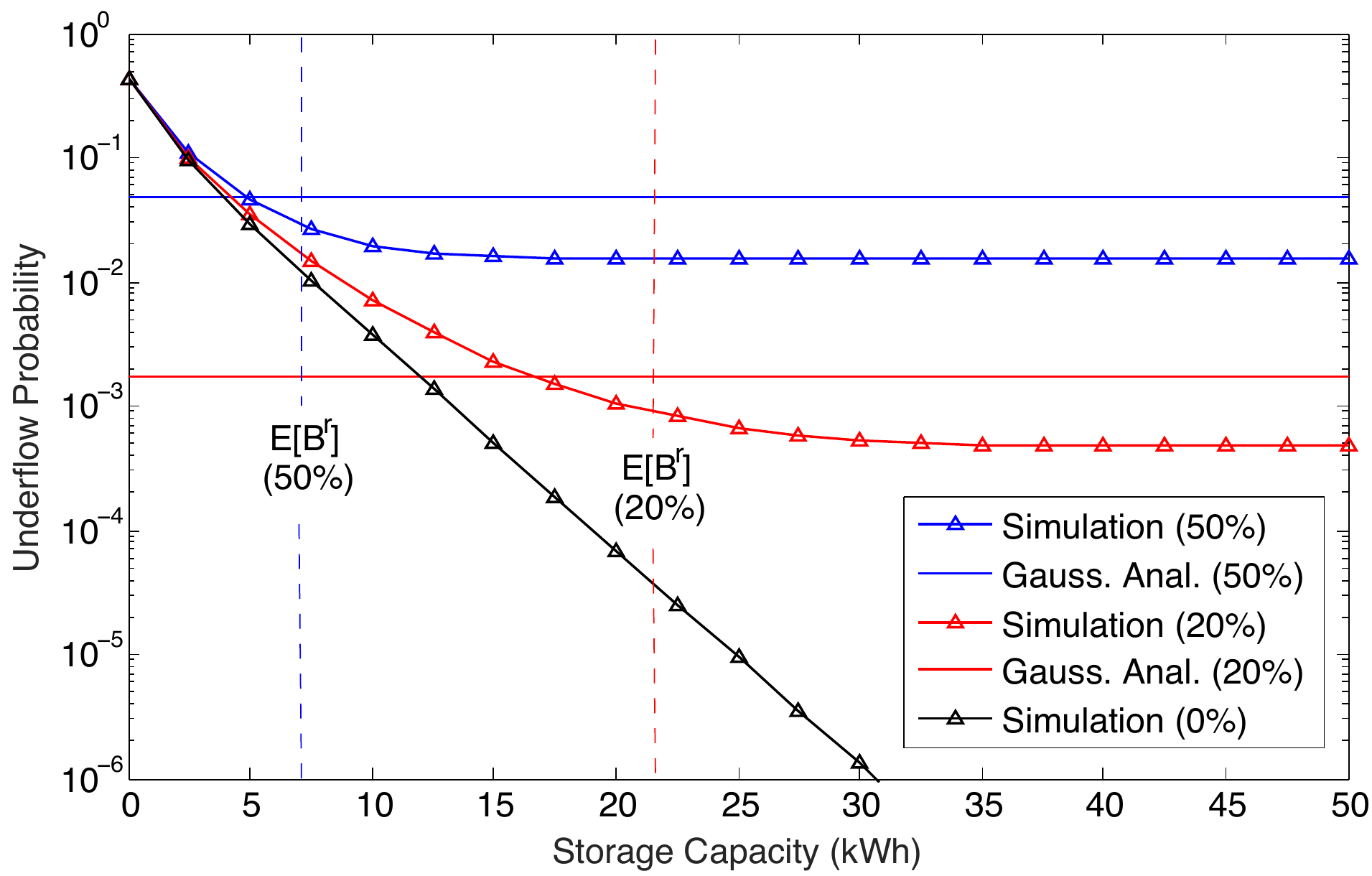}
	\label{fig:P_underflow_gauss_leakdom}
	}

\centering
	\subfigure[Energy Waste (Overflow).]{
 	\includegraphics[width=5in]{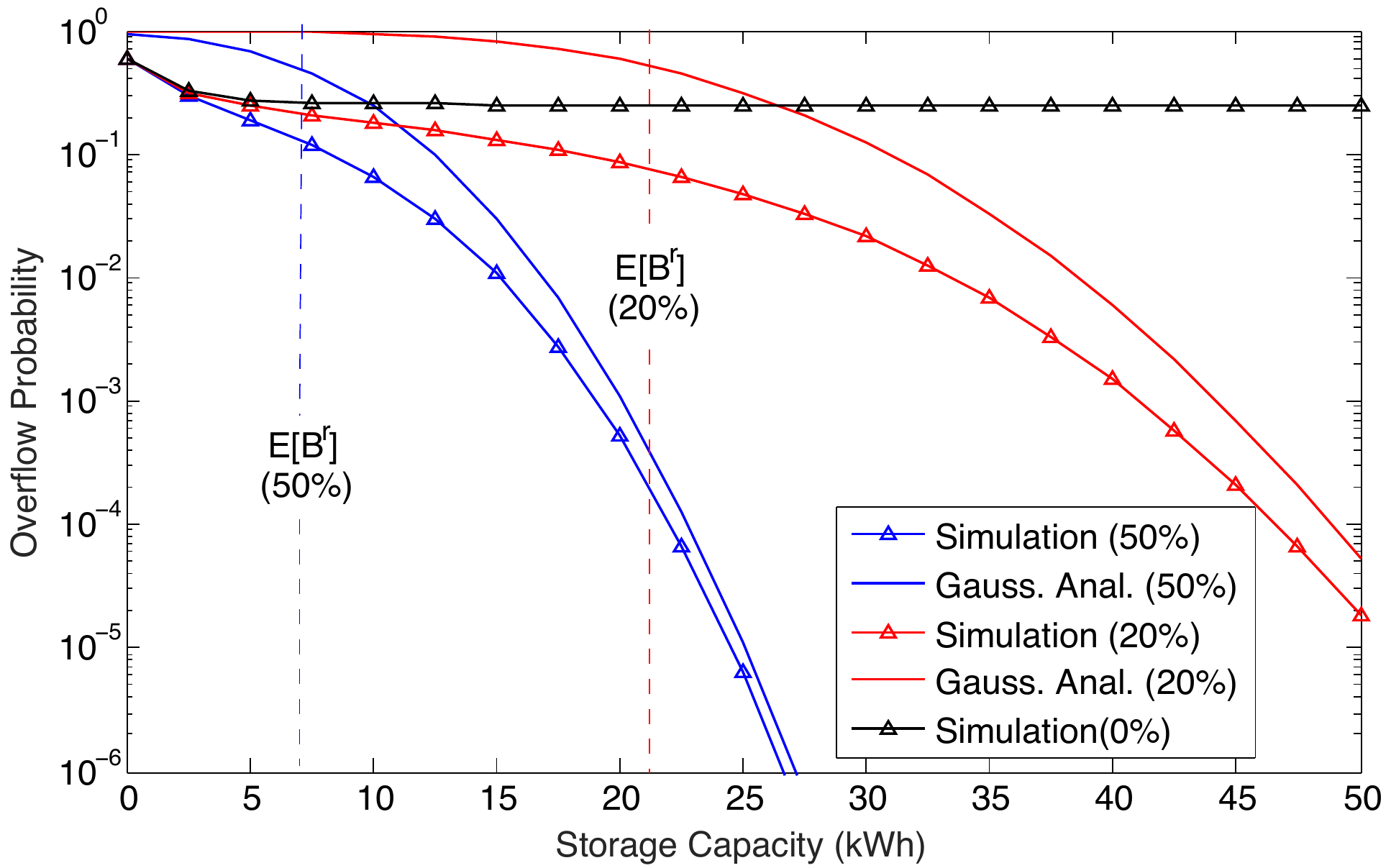}
	\label{fig:P_overflow_gauss_leakdom}
	}
  \caption{Underflow and overflow with Gaussian net charge 
{\rm   ($\E[\delta]=200$~Wh, $\sigma_{\delta}=1$~kWh, 
self-discharge per day: 0\%, 20\%, or~50\%)}. }
  \label{fig:gauss_leakdom}
\end{figure}

\clearpage 
In Fig.~\ref{fig:P_overflow_gauss_leakdom} we consider the 
wasted energy due to overflows.
For the system without leakage, the overflow probability 
quickly settles at a value 
that does not depend on the storage capacity. Here, the 
storage is mostly full and 
the overflow compensates for the excess 
supply compared to the demand. 
Leakage queues with non-zero leakage show a dramatically different behavior. 
The analytical estimate for $\P (l_o > 0)$ from Eq.~\eqref{eq:gauss_approx-lo} decreases 
faster than exponentially in $C$, which is also reflected in 
the simulations.  
We conclude that in leakage queues the overflow probability can be reduced 
arbitrarily by increasing the storage capacity.

Both plots in Fig.~\ref{fig:gauss_leakdom} show that the Gaussian analysis 
can provide good estimates of the underflow and overflow probabilities when 
the system is in the leakage-dominated regime ($C > \E[B^r]$).  
Even if $\delta(n)$ does not follow a normal distribution, 
the Gaussian analysis provides good estimates 
(see Subsec.~\ref{subsec:wind}). 
To understand why this is the case, let us 
consider general supply and demand processes that are i.i.d. with 
arbitrary distributions.
By Eq.~\eqref{eq:b_r_inf}, the 
stored energy $B^r$ in the steady state
of the reference system
is the sum of independent random variables $\delta(m)\bar\gamma^m$.
In the limit $\gamma\to 0$, these random variables become i.i.d.
In analogy with the Central Limit Theorem, 
one should expect that the distribution of $B^r$ 
approaches a Gaussian. The following result states that
a suitably normalized version of
$B^r$ converges to the normal distribution.

\begin{theorem}
Let $\delta(n)$ be a sequence of i.i.d. random variables with finite
mean and variance. For $0<\gamma<1$,
let $B^r$ be given by Eq.~\eqref{eq:b_r_inf}. Then, as $\gamma \to 0$, 
$$
Z:= \frac{B^r -\E[B^r]}{\sqrt{\Var [B^r]}}
$$
converges in distribution to a standard normal random variable.
\label{thm:gaussian}
\end{theorem}

\begin{proof} It suffices to show that the characteristic function
$\E[e^{i\theta Z}]$ converges to  $e^{-\theta^2/2}$,
the characteristic function of the standard normal
distribution, for every $\theta\in\RR$ \cite[Theorem 3.3.6]{Durrett}.
Let $\mathcal{X} (\theta)$ be the characteristic
function of the normalized random variable
$\frac{\delta-\E[\delta]}{\sqrt{\Var[\delta]}}$.
By the i.i.d. assumption,
$$
\E[e^{i\theta  Z}]
= \prod_{k=0}^\infty \mathcal{X} 
\left( \theta\sqrt{1-\bar\gamma^2}\bar\gamma^k\right)\,.
$$
Since $\delta$ has finite variance,
$\mathcal{X}$ is twice differentiable at zero,
with $\mathcal{X}(0)=1$, $\mathcal{X}'(0)=0$, and $\mathcal{X}''(0)=1$. 
Using the Taylor expansion of $\mathcal{X}$ about zero,
a routine estimate for the product
(see, for example,~\cite[Exercise 3.1.1]{Durrett}) 
shows that $\E[e^{i\theta  Z}] \to e^{-\theta^2/2}$
as $\gamma\to 0$.
\end{proof}

The theorem, in combination with Corollary~\ref{coro:B-Br},
justifies Eq.~\eqref{eq:gauss_approx-lu} and
Eq.~\eqref{eq:gauss_approx-lo} for general i.i.d. arrival 
and service processes of finite mean and variance,
so long as $\gamma$ is sufficiently small 
and $C > \E[B^r]$.
If $\gamma$ is not close to zero, additional moments
of $B^r$ may be used to approximate its distribution.
Since $B^r$ is a sum of independent random variables,
each of its cumulants can be computed directly from the corresponding
cumulant of $\delta$, using Eq.~\eqref{eq:b_r_inf}. 
In particular, the skewness of $B^r$ is given by
 \begin{align}
\text{Skew}[B^r] & =\frac{(1-\bar\gamma^2)^{3/2}}{1-\bar\gamma^3}
\text{Skew}\left[\delta\right]\,.
\label{eq:skew-normal}
 \end{align}
For non-zero skewness, a skew-normal distribution,
fitted to the mean, variance, and skewness of $B^r$,
provides a better approximation to
$\P(B^r>C)$ and $\P(B^r<0)$ than a Gaussian~\cite{skew-normal}.
This approximation can
be inserted in place of the Gaussian on the right
hand sides of Eq.~\eqref{eq:gauss_approx-lu} and
Eq.~\eqref{eq:gauss_approx-lo}.  
We will illustrate the benefits of using the skew normal approximation
in Subsec.~\ref{subsec:wind}.

\subsection{Martingale analysis}\label{subsec:bounds-martingale}

We have seen that the Gaussian analysis from the 
previous subsection provides a good level 
of accuracy in the leakage-dominated regime, but much less so 
in the capacity-dominated regime. 
Next, we present a martingale analysis 
for estimating the underflow probability in the capacity-dominated regime.
Note that, in the capacity-dominated regime, the overflow probability is of less 
interest since the stored energy is mostly close to or at 
capacity.

Recall that at $\E[\delta]>0$,  a leakage queue in the capacity-dominated 
behaves similarly to a finite-capacity queueing system.
In the queueing 
literature, the overflow probability  of finite-capacity queues has been 
studied extensively.  We take advantage of the available methods 
for buffer overflow by resorting to 
the dual system presented in Subsec.~\ref{subsec:duality} 
and applying the overflow analysis in \cite{kobayashi}, 
which is based on the  Kingman-Ross delay 
bounds~\cite{kingman,ross-martingale}. 

\begin{theorem}
Consider a leakage queue with a net charge 
given by an i.i.d. process
$\delta(n)$ and leakage ratio $0<\gamma<1$. 
Assume that the moment-generating function
$M(\theta)=\E[e^{\theta \delta}]$ exists
for at least some $\theta<0$,  and let
$$
\theta^*=\sup\{\theta >0: \E[e^{\theta^*(\gamma C -\delta)}]\le 1\}\,.
$$
If $C < E[B^r]$, then $\theta^*>0$ and
the underflow probability in the steady 
state is  bounded  by 
\begin{equation}
\P (l_u > 0 ) \leq e^{-\theta^*C}\,.
\label{eq:regime1_underflow}
\end{equation}
If, moreover, $a(n)$ and $s(n)$ are independent processes, then
\begin{equation}
\P (l_u > 0 ) \leq \frac{e^{-\theta^*C}}{
\inf_{x} \E[e^{\theta^*(s-x)}\vert
s >x ]}\,.
\label{eq:regime1_underflow-imp}
\end{equation}
\label{thm:martingale}
\end{theorem}

We will use the theorem in the special case where
$s(n)=s_0 + s_1(n)$, where $s_0$ is a constant 
and $s_1(n)$ is exponentially distributed.
In that case, the denominator of \eqref{eq:regime1_underflow-imp}
can be evaluated exactly, resulting in 
\begin{align*}
\P (l_u > 0 ) \leq e^{-\theta^*C}(1-\theta^*\E[s_1])\,.
\end{align*}
The following proof extends a recent martingale analysis of a fork-join 
system~\cite{Ciucu-forkjoin} to a leakage queue. 
\begin{proof}
Let $B'(n)$ be the dual system introduced
in Subsec.~\ref{subsec:duality}.
Let $\delta'(n)$ be the net charge in the dual system, 
as defined in Subsec.~\ref{subsec:duality}, 
and let $M'(\theta)=\E[e^{\theta\delta'}]$
be its moment-generating function.
Since $M'$ is a convex function, 
the set $\{\theta\in\RR\vert M'(\theta)\le 1\}$
is an interval containing zero. 
By definition, $\theta^*$ is the right endpoint of that interval.
The assumption that $\E[B^r]>C$ implies,
by Eq.~\eqref{eq:gauss_prob_1}, that 
$$
\E[\delta']=\gamma C -\E[\delta]<0\,.
$$
It follows that $\theta^*>0$.

Suppose the dual system is started with $B'(0)=0$.  
By Theorem~\ref{thm:backlog}, we have for every $n>0$,
\begin{align*}
B'(n) &=\min_{0\le m\le n} \bigl\{
\max_{m\le j\le n} \{
C \bar\gamma^{n-m}\mathbbm{1}_{j=m>0}
+ \Delta'_\gamma(j,n)\}\bigr\}\\
&\le \max_{0\le j\le n} \{\Delta'_\gamma(j,n)\}\\
&=_{\mathcal{D}} \max_{0\le m\le n} \left\{
\sum_{k=0}^{m-1} \delta'(k)\bar \gamma^k\right\}\,,
\end{align*}
where
$$
\Delta'_\gamma(m,n)=\sum_{k=m+1}^n \delta'(k)\bar\gamma^{n-k}\,.
$$
In the second line, we have set $m=0$, and in the third
line, we have used the i.i.d. assumption.
Taking $n\to\infty$, we obtain for the steady state
$$
\P(l_u>0) = \P(l_o'>0)
\le \P\left(\max_{m\ge 0}
\sum_{k=0}^{m} \delta'(k)\bar \gamma^k >C\right)\,.
$$
Note that the steady state does
not depend on the choice of $B(0)$
by Theorem~\ref{thm:convergence}.

Define for $m\ge 0$
\begin{align*}
z(m)&=e^{\theta^* \sum_{k=0}^m {\delta'(k) \bar \gamma^{k}}}\,.
\end{align*}
Then 
\begin{align*}
z(m) & =e^{\theta^* \delta'(m)\bar\gamma^m}z(m-1) \, , \\
z(0) & =e^{\theta^* \delta'(0)} \, . 
\end{align*}
Since $\delta'(m)$ is independent of $z(m-1)$, we have
\begin{align*}
\E[z(m)\mid z(m-1)]
&= \E \bigl[e^{\theta^* \delta'(m)\bar\gamma^m}\bigr] \, z(m-1)\\
& \le z(m-1)\,.
\end{align*}
Therefore, $z(m)$ is a supermartingale.
By Doob's inequality for positive supermartingales,
\begin{align}
\notag
\P(l_u>0)
&\le \P\left(\max_{m\ge 0} z(m) > e^{\theta^* C}\right)\\
\notag &\le \P\left(z(0) > e^{\theta^* C}\right)\\
&\le e^{-\theta^* C}\,.
\notag 
\end{align}
In the last step, we have used Markov's inequality
and the fact that $\E[z(0)]=1$.
This proves the claim in Eq.~\eqref{eq:regime1_underflow}.

The bound can
be sharpened by a stopping-time argument, as follows.
For $k\ge 0$, let $t_k$ be 
the minimum of $k$ and the number of the
first time slot where $z$ exceeds $e^{\theta^* C}$.
Clearly, each $t_k$ is a random variable of finite expectation,
and $t_{k}\le t_{k+1}$ for all $k$.
Since $\max_{m\ge 0} z(m)>e^{\theta^* C}$ if and
only if $\lim_{k\to\infty} t_k<\infty$,
it follows that
\begin{equation}
\P(l_u>0)\le \lim_{k\to\infty} \P(t_k<k)\,.
\label{eq:proof_k}
\end{equation}
By the optional stopping theorem and the positivity
of $z(m)$,
\begin{align*}
1&=\E[z(0)] \\
&\geq \E[z(t_k)]\\
& \geq \E[z(t_k)\mid t_k<k] \,\P(t_k<k)
\end{align*}
for each $k\ge 0$, that is,
\begin{equation}
\P(t_k<k)\le \frac{1}{\E[z(t_k)\mid t_k<k]}\,.
\label{eq:proof_cond_e}
\end{equation}
Now, using the method of~\cite{ross-martingale} we find that
 \begin{align*} 
\E\bigl[z(t_k)\mid t_k<k\bigr] &=\E\bigl[z(t_k)\mid z(t_k)>e^{\theta^*C}\bigr]\\
 &=e^{\theta^* C} 
\E\Bigl[e^{\theta^*\left( \sum_{j=0}^{t_k} {\delta'(j) \bar\gamma^{j}}-
C\right)}\Big\vert \sum_{j=0}^{t_k} 
\delta'(j) \bar\gamma^{j}> C\Bigr]\nonumber \\
 &\ge e^{\theta^* C} \inf_{x\in\RR} 
\E\bigl[e^{\theta^* (\delta'(0)-x)}
\, \bigl| \, \delta'(0)>x\bigr]\,, 
\end{align*}
where the last line follows by conditioning on
$t_k$ and on the value of $C-\sum_{j=1}^{t_k} \delta'(j)$.
If $a(n)$ and $s(n)$ are independent  
one can condition also on $a(0)$ to obtain 
\begin{equation}
\E\bigl[z(t_k)\, \bigl| \, t_k<k\bigr]
\ge e^{\theta^*C}
\inf_{x\in\RR}\E\bigl[e^{\theta^* (s-x)}
\, \bigl| \, s>x\bigr]\,.
\label{eq:proof_inf_e}
\end{equation}
The proof of Eq.~\eqref{eq:regime1_underflow-imp}
is concluded by inserting Eq.~\eqref{eq:proof_inf_e}
into Eq.~\eqref{eq:proof_cond_e} and then
using Eq.~\eqref{eq:proof_k}.
\end{proof}

The theorem does not always provide accurate estimates in the 
capacity-dominated regime. The bounds on the loss probability 
are acceptable when $\gamma$ is close to 
zero and the randomness of the system is limited. 

 \begin{figure}[!t]
\centering
\includegraphics[width=5in]{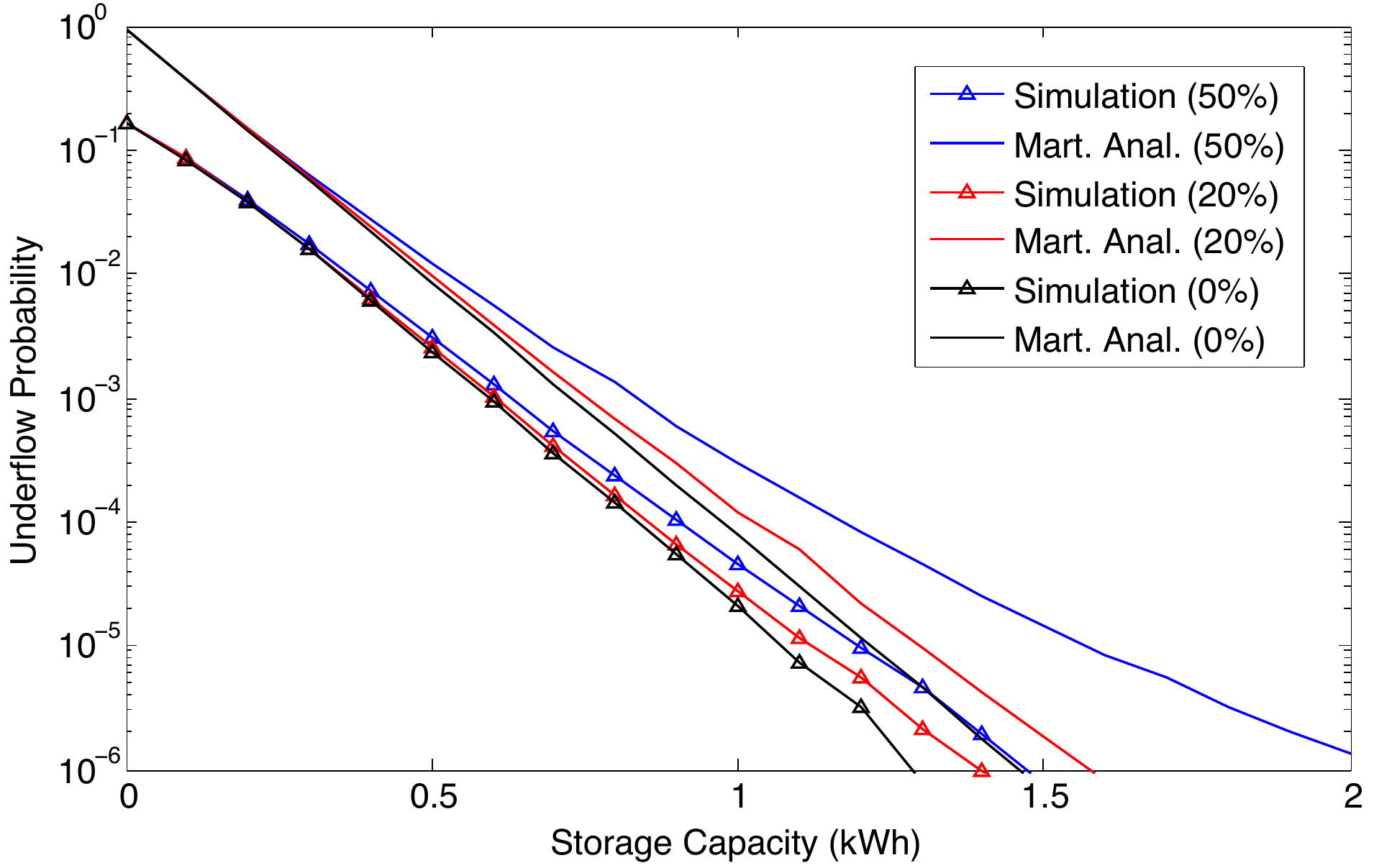}
  \caption{Underflow in the capacity-dominated regime. 
{\rm ($\E[\delta]=200$~Wh, $\sigma_{\delta}=50\sqrt{2}$~Wh 
and self-discharge per day of $0\%$, $20\%$, or $50\%$)}.}
\label{fig:p_underflow_gauss_buffdom}
\end{figure}

In Fig.~\ref{fig:p_underflow_gauss_buffdom}, we evaluate 
Theorem~\ref{thm:martingale} for a  leakage queue with Gaussian 
supply and demand processes as in Subsec.~\ref{subsec:bounds-gaussian}. 
We consider a storage capacity of $\le 2$~kWh, 
which is clearly in the capacity-dominated 
regime. For the parameters used in Subsec.~\ref{subsec:bounds-gaussian}, 
the martingale bound does not result in good estimates. 
The results improve when we reduce the randomness of the stored energy process.  
In Fig.~\ref{fig:p_underflow_gauss_buffdom} we have done this 
by reducing  the standard deviations to 
$\sigma_a=\sigma_s = 50$~Wh, resulting in 
$\sigma_\delta =50\sqrt{2} $~Wh.
In this case, the analytical results provide useful upper bounds on the 
underflow probabilities. Note that a storage system in the capacity-dominated 
regime with a high degree of randomness will result in high overflow probabilities as well as high underflow probabilities. While it is not apparent that such a parameter region is of interest for deployed energy storage systems, analytical methods  
that provide good bounds in this region remain an open problem.

\section{Evaluation of a Wind Energy Model}
\label{sec:numerics}

\begin{figure*}[ht]
  \centering
	\subfigure[Energy Loss (Underflow).]{
 	\includegraphics[width=5in]{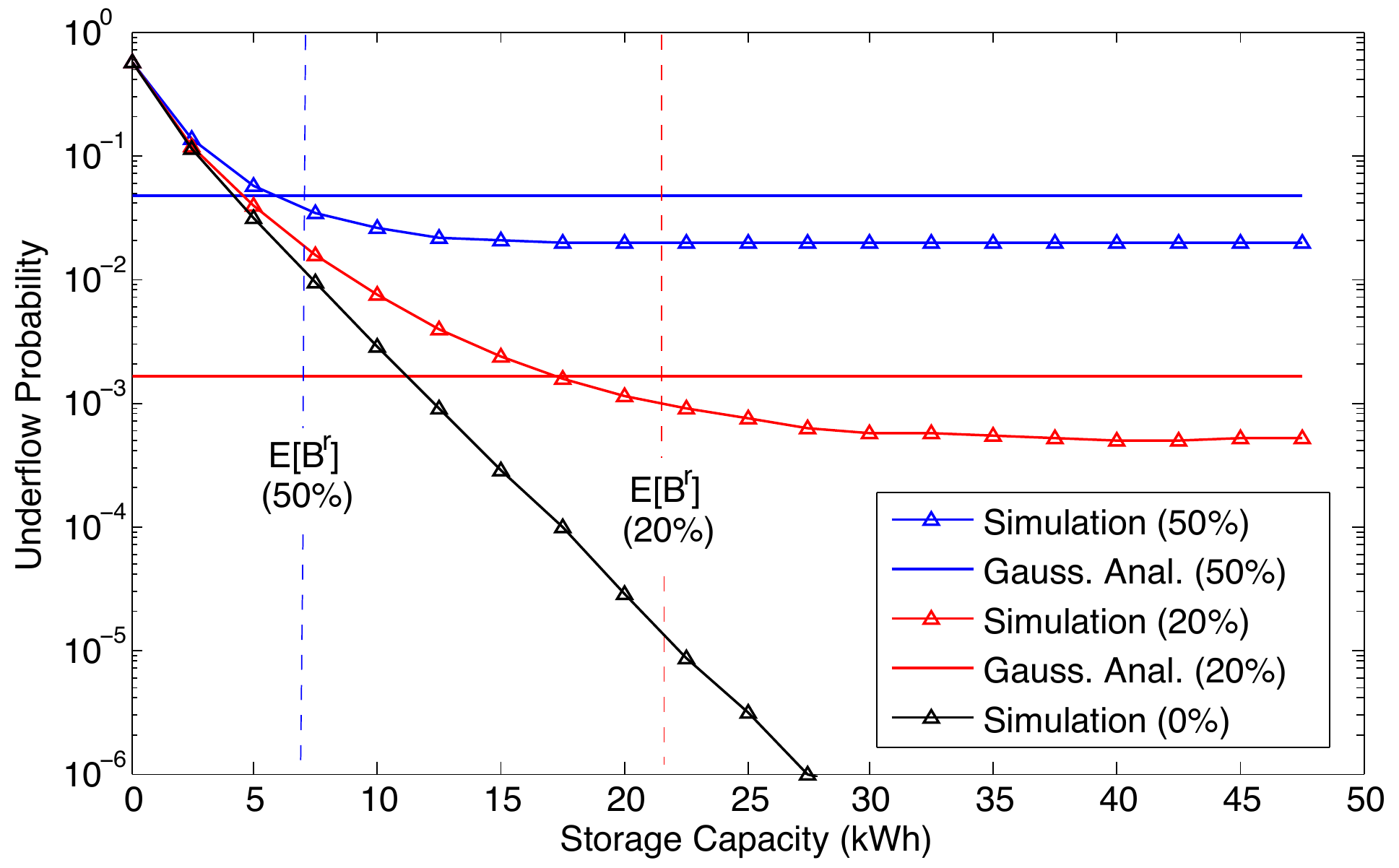}
	\label{fig:P_underflow_wind_leakdom}
	}
	\subfigure[Energy Waste (Overflow).]{
 	\includegraphics[width=5in]{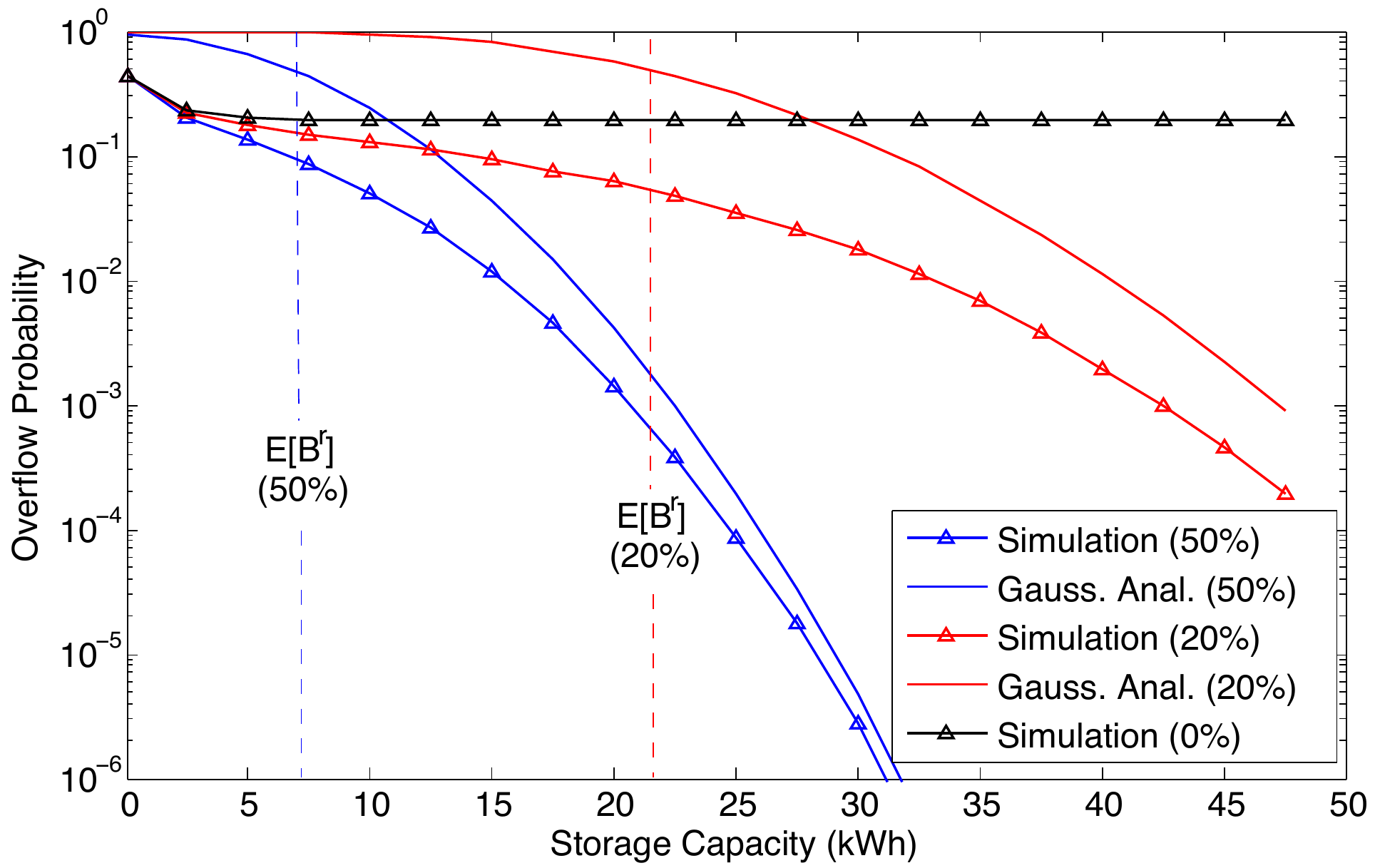}
	\label{fig:P_overflow_wind_leakdom}
	}   
  \caption{Underflow and overflow probabilities with wind energy source. 
{\rm ($\E[a]=1$~kWh, $\E[s]=800$~Wh, $\sigma_a=1050$~Wh, 
$\sigma_s=50$~Wh.  The self-discharge per day is 0\%, 20\%, or 50\%)}.}
 \label{fig:p_wind_leakdom}
\end{figure*}
We next consider a leakage queue 
with a supply process that resembles a wind energy source. 
Our objectives are twofold. First, we want to see if the reference system 
remains useful in the context of more realistic, and more complex, random processes. 
Second, we want to evaluate the accuracy of our analysis.  
We use the wind speed process from \cite{mod_est}, 
which models wind speed as an i.i.d. process with  a Weibull distribution with 
density function
\begin{equation*}
f_V(v)=\frac{k}{c} \left(\frac{v}{c}\right)^{k-1} 
e^{-{\left(\frac{v}{c}\right)}^k} \, ,
\end{equation*}
where $c$ and $k$, respectively, are the scale and shape parameters of the 
Weibull distribution.
As in~\cite{mod_est}, we set the shape parameter to $k=3$ in order to bound 
the degree of randomness of the wind speed. 
The scale parameter factor is set to $c=7$~m/s,  
which results in an average wind speed of $6.25$~m/s. 
 
We consider a wind turbine with a rated power of $P_r = 1$~kW, 
which is comparable to a micro wind turbine for a residential 
home~\cite{microturbine}. 
The  output power of wind turbines, denoted by  $P_w$ and expressed in 
$kW/\text{m}^2$ is a function of the wind speed $v$. 
Wind turbines are activated only when 
the wind speed is above a lower threshold (cut-in speed) 
and below an upper threshold (cut-out speed).  
The rated speed is the wind speed at which the 
wind turbine generates its rated power~$P_r$. 
Using the power model from~\cite{dec_sup}, we obtain 
\begin{align}
P_w=
\begin{cases}
0 &  v<v_{ci} \, , \\
\alpha v^3-\beta P_r & v_{ci}\leq v \leq v_r \, , \\
P_r &  v_r\leq v\leq v_{co} \, , \\
0 &  v_{co}\leq v \, .  
\end{cases} \label{eq:power_curve}
\end{align}
where $v_{ci}$, $v_r$ and $v_{co}$, respectively, are the 
cut-in, rated, and cut-out wind speeds,  and $\alpha$ 
and $\beta$ are calculated such that 
Eq.~\eqref{eq:power_curve} is continuous at $v_{ci}$ and 
${v_r}$, i.e., $\alpha=\tfrac{P_r}{v_r^3-v_{ci}^3}$ and 
$\beta=\tfrac{v_{ci}^3}{v_r^3-v_{ci}^3}$. 
The actual power 
from the wind turbine is given by~\cite{dec_sup}
\[
a(n)=P_w\cdot A_w \cdot \eta_w,
\] 
where $A_w$ and $\eta_w$ are the total swept area and the 
efficiency of the wind turbine, respectively.
The parameters of the wind turbine are summarized in Table~\ref{ta:sim_table}. 
With these parameters we obtain a supply process with 
$\E[a] = 1$~kWh and $\sigma_a = 1050$~Wh. 

\begin{table}[h]
\centering
\caption{Parameters of the wind turbine.}
\begin{tabular}[t]{llr}
\hline
$\text{Notation}$ & Definition & value\\ 
\hline
$P_r$ & Rated power & 1~kW\\
$v_{ci}$ & Cut-in wind speed & 3~m/s\\
$v_{co}$ & Cut-out wind speed & 25~m/s\\ 
$v_r$ & Rated wind speed & 12~m/s\\ 
$A_w$ & Total swept area & 10.8~$\text{m}^2$\\
$\eta_w$ & Wind turbine efficiency & 50\%\\
\end{tabular}\label{ta:sim_table}
\end{table}

The demand process is set as the sum of a constant demand of 
$750$~Wh and an i.i.d. exponential random value with average $50$~Wh, 
resulting 
in $\E[s] = 800$~Wh and $\sigma_s = 50$~Wh.
We consider energy storage systems with a significant self-discharge, 
with leakage ratios $\gamma = 0.0093$ (20\% per day) and $\gamma = 0.0285$ 
(50\% per day), which is within the range of supercapacitors or flywheels. 
As before, we use~$1$~hour for the length of a time slot.

\newpage 
 
\subsection{Impact of self-discharge}
\label{subsec:wind}

In Fig.~\ref{fig:p_wind_leakdom}, we show the underflow and overflow probability 
as a function of 
the size of the energy storage system for different leakage ratios.  
We depict the results of the Gaussian analysis and the simulations. 
Note that the average net charge $E[\delta]$ as well as 
the leakage ratios match the examples in Sec.~\ref{sec:bounds}.  
Thus, by comparing Fig.~\ref{fig:p_wind_leakdom} with 
Fig.~\ref{fig:gauss_leakdom} we can gauge the impact of 
replacing  the Gaussian  energy supply and demand  with 
distributions that are quite different from a normal distribution. 
The results of the analysis and the simulations corroborate our earlier observations 
from Subsec.~\ref{subsec:bounds-gaussian}. 
The underflow and overflow probabilities of a 
leakage queue are significantly  different from 
those of a conventional finite-capacity queueing system ($\gamma = 0$). 
Different from a conventional queue, the underflow probability of a leakage 
queue does not vanish when the storage capacity is increased. 
The reverse is true for the overflow probability. 
As seen in Subsec.~\ref{subsec:bounds-gaussian}, the conventional queueing system is generally at capacity with (an 
eventually) constant overflow probability. 
As for the accuracy of the Gaussian analysis, 
the estimates of the underflow and overflow probabilities are good 
in the leakage-dominated regime ($C > E[B^r]$). 
We conclude that the accuracy of the Gaussian analysis does not 
deteriorate when moving to non-Gaussian distributions for supply and demand. 

\clearpage
\begin{figure}[!t]
\centering
\includegraphics[width=5in]{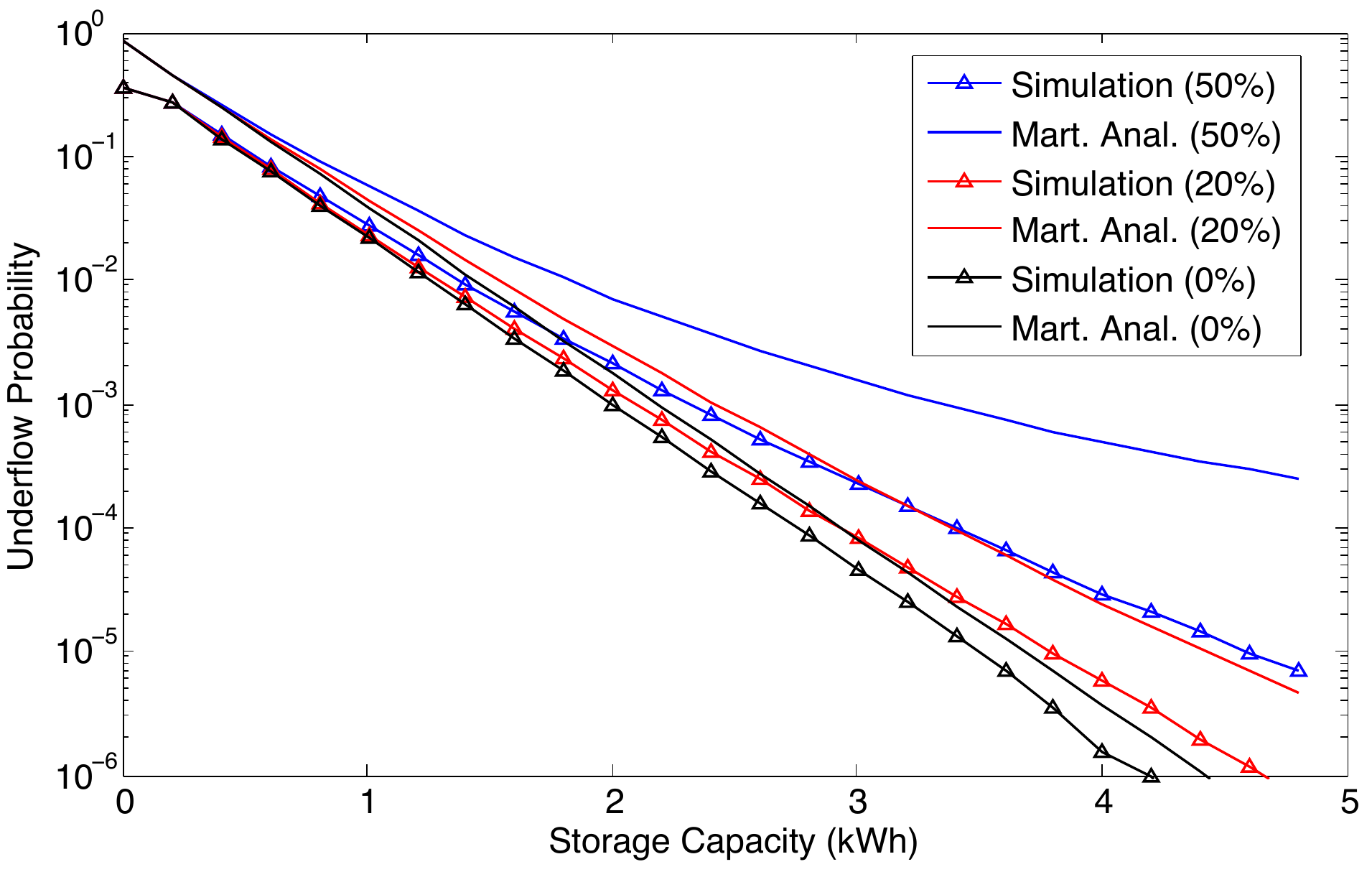}
  \caption{Underflow probability with wind energy in the capacity-dominated regime. 
{\rm ($\E[a]=1$~kWh, $\E[s]=400$~Wh, $\sigma_a=1050$~Wh, $\sigma_s=50$~Wh. 
The self-discharge per day is 0\%, 20\%, or 50\%)}.}
\label{fig:P_underflow_wind_buffdom}
\end{figure}

In Fig.~\ref{fig:P_underflow_wind_buffdom},  we present the underflow probabilities 
computed with the martingale analysis for a capacity-dominated regime where $C \le 5$~kWh. 
Recall that, in this range,  the overflow probability is always close to one. 
Since the martingale method does not provide usable bounds 
unless the underflow and the overflow probabilities are small, 
we reduce the randomness of the stored energy process $B (n)$ 
by reducing the demand process to $\E[s]=400$~Wh (and keeping $\sigma_s = 50$~Wh.)
Note that the expected values are selected to match those used in 
Fig.~\ref{fig:gauss_leakdom}. Comparing 
Fig.~\ref{fig:P_underflow_wind_buffdom} and Fig.~\ref{fig:gauss_leakdom}, we see 
that the accuracy of the martingale analysis is comparable or even improved.

\subsection{Validation of the reference system}
We return to the main finding of this paper, which is the 
distinct behavior of the leakage queue in the leakage-dominated 
regime. 
In Subsec.~\ref{subsec:ref_system}, we used empirical 
data to show that the reference system provides an accurate 
characterization of the leakage queue in the leakage-dominated regime. 
Further, in Subsec.~\ref{subsec:bounds-gaussian}, we proved that the 
reference system approaches a normal distribution when~$\gamma$ is small. 
Our final evaluation validates that the leakage queue 
in the leakage-dominated regime is well described by a normal 
distribution even if the supply and demand processes are not Gaussian. 

We work with the same supply and demand processes described at the 
beginning of this section. We consider energy storage systems with 
$\gamma=0.0093$ (20\% self-discharge per day), and, therefore,  
$\E[B^r] = 21.5$~kWh.
Fig.~\ref{fig:CDFplot_wind_fixedGamma20perc} presents the 
CDF for the stored energy 
obtained from simulations of the leakage queue with 
capacities set to $C=10, 20, 40$~kWh. 
The  queue is leakage-dominated if 
$C=40$~kWh, and capacity-dominated otherwise. 
We compare the distribution of the stored energy of the Gaussian 
approximation of the reference system 
(using Eq.~\eqref{eq:gauss_approx-lu})
with those obtained 
for a simulated leakage queue. 
We observe that for $C=40$~kWh, the simulation closely follows 
the normal distribution of the reference system. 
In the capacity-dominated regime ($C=10, 20$~kWh),  
the distributions of the stored energy do not resemble at all  a normal 
distribution. 

\medskip
The characteristics  of the distributions of the stored energy become more evident 
in a quantile-quantile (Q-Q) plot, where we compare the quantiles of 
the normal distribution with those of the simulations in 5\% increments. 
In Fig.~\ref{fig:Q_Qplot_wind_fixedGamma20perc},  we provide the quantiles of 
the normal distribution on the horizontal axis. The diagonal line, shown as a 
thin  solid line, therefore depicts the normal distribution. 
The thick solid line (in gray) has the results for the skew normal 
distribution, using the skewness for $B^r$ from~Eq.~\eqref{eq:skew-normal}. 
The values obtained from the simulations are shown as colored data points. 
The closer the data points are to the diagonal, the better the match is with the 
normal distribution. 
We observe that the stored energy in a leakage-dominated regime ($C=40$~kWh) 
is very close to the diagonal. The match is further improved with the skew 
normal distribution. 
The capacity-dominated regime ($C=10, 20$~kWh) is obviously 
poorly matched with a normal distribution. Even, for  $C=20$~kWh, when the storage 
capacity is close to $\E[B^r]$, the Q-Q plot is far from the diagonal. 


\begin{figure}[!t]
  \centering
	\subfigure[CDF.]{
 	\includegraphics[width=5in]{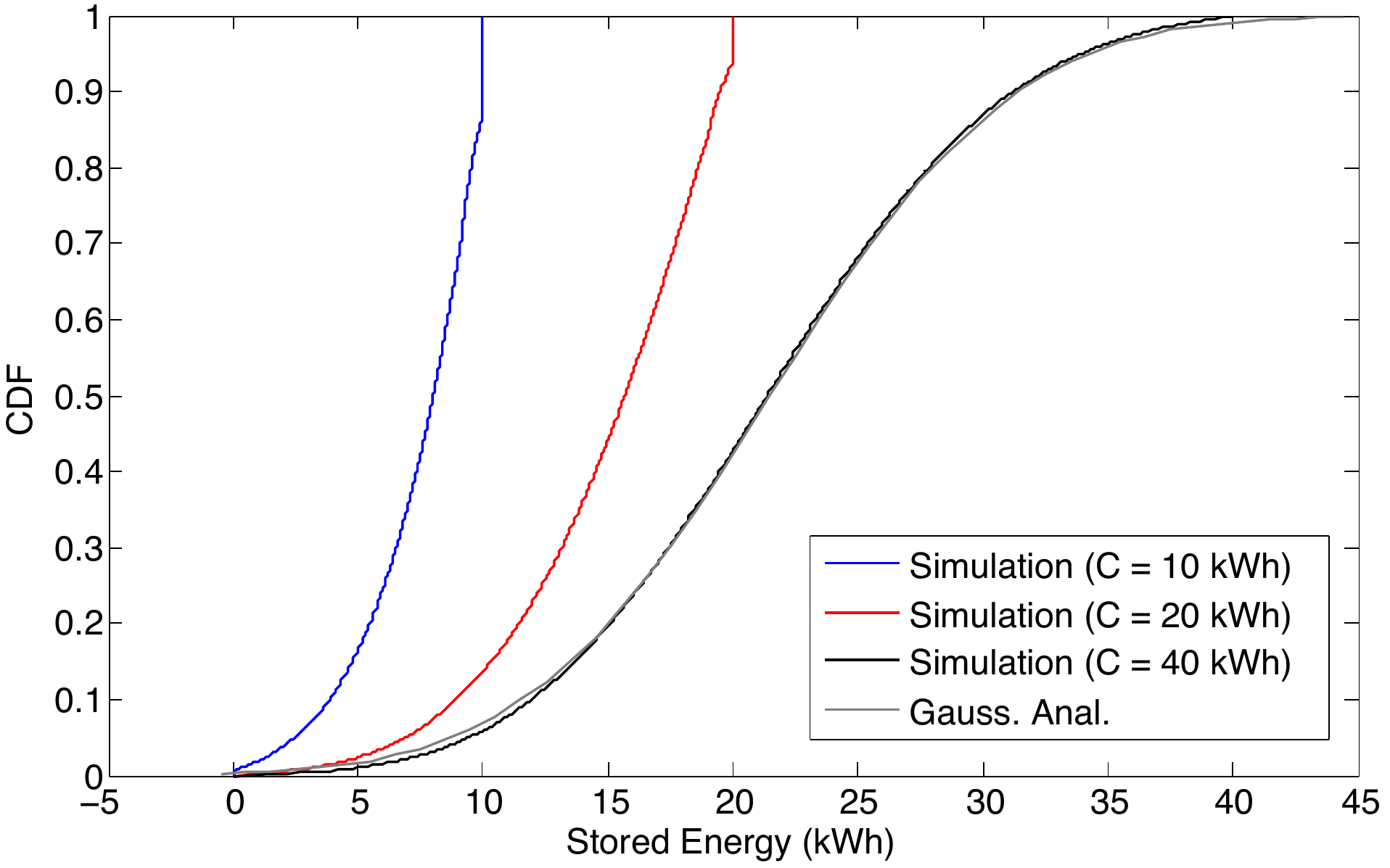}
	\label{fig:CDFplot_wind_fixedGamma20perc}
	}

\vspace{6pt}
\centering
	\subfigure[Q-Q plot.]{
 	\includegraphics[width=5in]{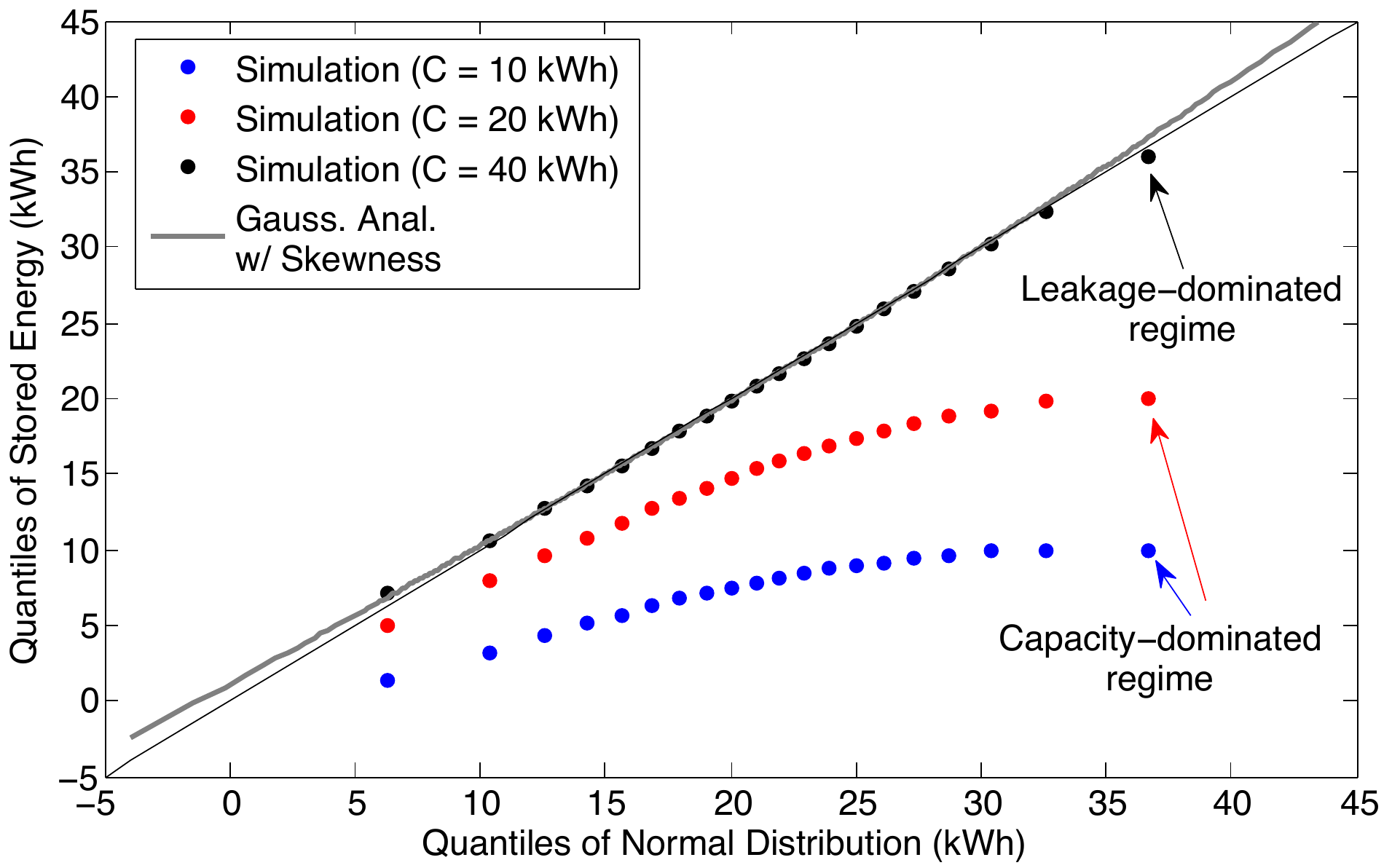}
	\label{fig:Q_Qplot_wind_fixedGamma20perc}
	}
\caption{Comparison of distributions of the stored energy in a leakage queue. 
{\rm ($\gamma$=0.0093,  $C=10, 20, 40$~kWh, 
$\E[a]=1$~kWh, $\E[s]=800$~Wh, $\sigma_a=$1050~Wh, $\sigma_s=$50~Wh).} }

\end{figure}

\clearpage
\section{Conclusions}
\label{sec:concl}

We presented an analysis of a queueing model for an energy 
storage system with self-discharge. The model, referred to as 
{\it leakage queue}, has, in addition to a supply and a demand process, 
a self-discharge process that  removes  storage content proportionally to the 
filling level. 
We identified two distinct parameter regimes for the leakage queue,
which we called   leakage-dominated regime and capacity-dominated regime. 
In the  leakage-dominated regime, the queue settles in a steady state  below the storage capacity. 
In the capacity-dominated regime, the leakage queue resembles  
a conventional finite-capacity queueing system. 
We presented analytical methods for computing probabilities of 
underflow and overflow, and evaluated their accuracy.  
Extensions of our work should address a relaxation of the discrete-time 
assumption to a continuous-time system as well as a relaxation of the i.i.d. assumptions to 
general stationary processes. 
The empirical solar irradiance data used in Subsec.~\ref{subsec:ref_system} suggests 
that the regimes of the leakage queue manifest themselves even for non-stationary 
supply processes. It is also desirable to 
improve the martingale analysis for systems where the randomness of the 
system is high, and both underflow and overflow events occur frequently. 
Other directions for future research are an experimental confirmation 
of our findings on leakage queues in an actual energy storage system 
with self-discharge. An interesting  question is how to take advantage of the observed effects of self-discharge  
to improve the design and control of energy storage systems. 

\section*{Acknowledgements} 
This work is supported in part by the Natural Sciences and Engineering Research Council of Canada (NSERC).


%


\end{document}